\documentclass[runningheads]{llncs}

\usepackage{microtype}

\usepackage{hyperref}
\hypersetup{colorlinks,linkcolor={black},citecolor={black},urlcolor={black}}

\usepackage{thmtools}
\usepackage{thm-restate}

\usepackage{amsfonts}
\usepackage{stmaryrd}        
\usepackage{booktabs}
\usepackage{nicefrac}
\usepackage[inline]{enumitem}
\usepackage{mathtools}
\usepackage[binary-units]{siunitx}
\usepackage{subcaption}

\usepackage{tikz}
\usetikzlibrary{positioning}
\usetikzlibrary{arrows,automata}
\usetikzlibrary{calc}
\usetikzlibrary{petri}
\usetikzlibrary{decorations.pathreplacing}
\usetikzlibrary{decorations.pathmorphing}
\usetikzlibrary{trees}
\usetikzlibrary{shapes,shapes.geometric,fit}
\usetikzlibrary{backgrounds}
\usepackage{complexity}

\usepackage{macros}

\usepackage[firstpage]{draftwatermark}
\SetWatermarkText{\includegraphics[width=0.20\textwidth]{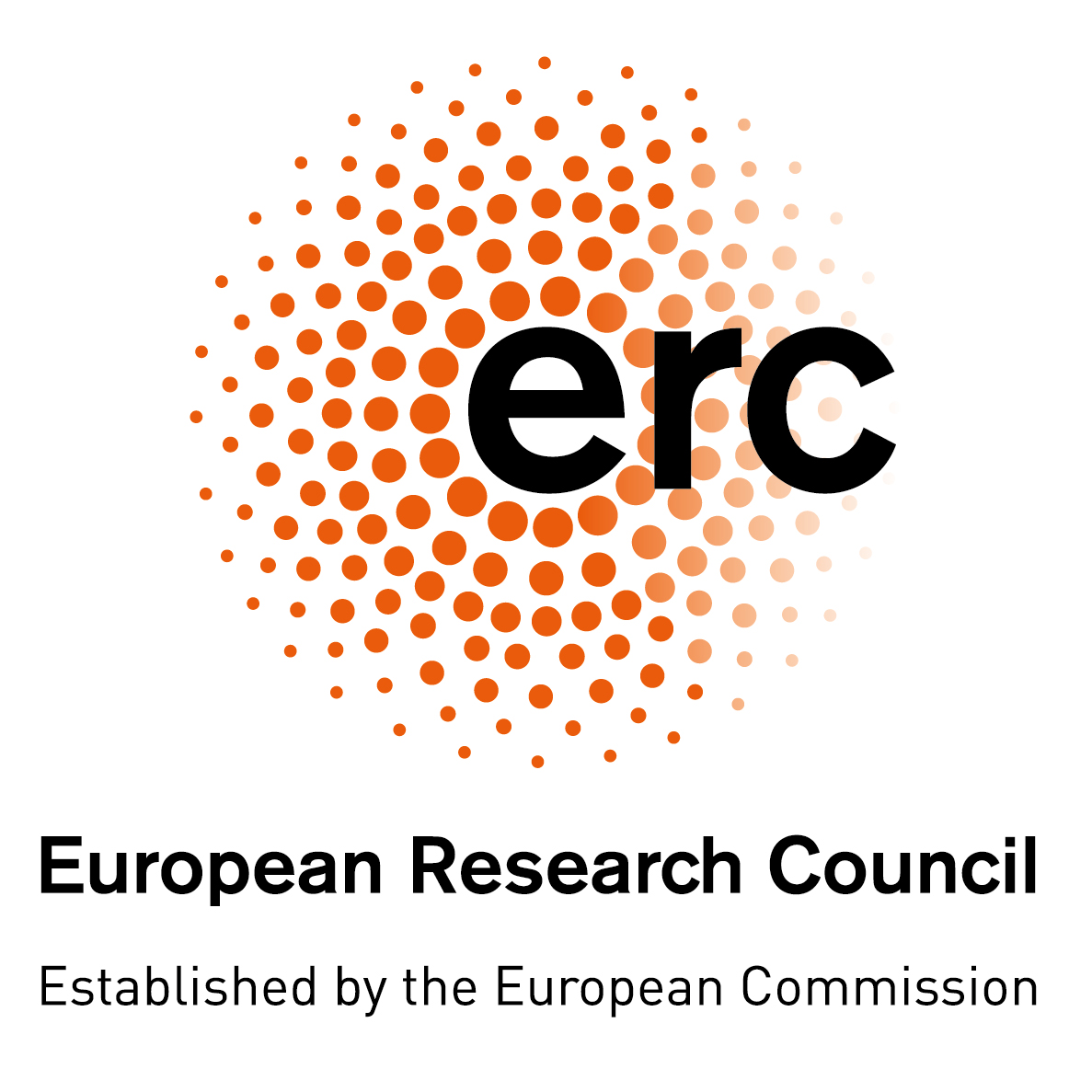}}
\SetWatermarkAngle{0}
\SetWatermarkHorCenter{0.725\paperwidth}
\SetWatermarkVerCenter{0.81\paperheight}

\begin{document}
\title{Computing the Expected Execution Time of Probabilistic Workflow Nets\thanks{%
\parbox[t]{0.76\textwidth}{%
The project has received funding from the European Research Council (ERC) under
the European Union’s Horizon 2020 research and innovation programme under grant
agreement No 787367 (PaVeS). Further it is partially supported by the DFG
Project No. 273811150 (Negotiations: A Model for Tractable Concurrency).
\strut}\hfill%
\parbox[t]{0.20\textwidth}{%
\vspace{-5pt}%
\strut}%
}}
%
%
\author{Philipp J. Meyer\inst{} \and
Javier Esparza\inst{} \and
Philip Offtermatt\inst{}}
\authorrunning{P. J. Meyer \and J. Esparza \and P. Offtermatt}
%
\institute{Technical University of Munich, Munich, Germany\\\email{\{meyerphi,esparza,offtermp\}@in.tum.de}}

\maketitle

\begin{abstract}
Free-Choice Workflow Petri nets, also known as Workflow Graphs, are a popular
model in Business Process Modeling.

In this paper we introduce Timed Probabilistic Workflow Nets (TPWNs), and give
them a Markov Decision Process (MDP) semantics. Since the time needed to
execute two parallel tasks is the maximum of the times, and not their sum, the
expected time cannot be directly computed using the theory of MDPs with
rewards. In our first contribution, we overcome this obstacle with the help of
``earliest-first'' schedulers, and give a single exponential-time algorithm for
computing the expected time.

In our second contribution, we show that computing the expected time is
\sharpp{}-hard, and so polynomial algorithms are very unlikely to exist.
Further, \sharpp{}-hardness holds even for workflows with a very simple
structure in which all transitions times are $1$ or $0$, and all probabilities
are $1$ or $0.5$.

Our third and final contribution is an experimental investigation of the
runtime of our algorithm on a set of industrial benchmarks. Despite the
negative theoretical results, the results are very encouraging. In particular,
the expected time of every workflow in a popular benchmark suite with 642
workflow nets can be computed in milliseconds.
\end{abstract}

\section{Introduction}\label{sec/introduction}
Workflow Petri Nets are a popular model for the representation and analysis of
business processes~\cite{DBLP:journals/jcsc/Aalst98,van2004workflow,DBLP:conf/bpm/DeselE00}.
They are used as back-end for different notations like BPMN (Business Process
Modeling Notation), EPC (Event-driven Process Chain), and UML Activity Diagrams.

There is recent interest in extending these notations with quantitative
information, like probabilities, costs, and time. The final goal is the
development of tool support for computing performance metrics, like the average
cost or the average runtime of a business process.

In a former paper we introduced Probabilistic Workflow Nets (PWN), a foundation
for the extension of Petri nets with probabilities and rewards~\cite{PEVA}.
We presented a polynomial time algorithm
for the computation of the expected cost of free-choice workflow nets, a
subclass of PWN of particular interest for the workflow process community (see
e.g.~\cite{DBLP:journals/jcsc/Aalst98,DBLP:journals/is/FavreFV15,FVM16,esparza2016reduction}).
For example, 1386 of the 1958 nets in the most popular benchmark suite in the
literature are free-choice Workflow Nets~\cite{FFJKLVW09}.

In this paper we introduce Timed PWNs (TPWNs), an extension of PWNs with time.
Following~\cite{PEVA}, we define a semantics in terms
of Markov Decision Processes (MDPs), where, loosely speaking, the
nondeterminism of the MDP models absence of information about the order in
which concurrent transitions are executed. For every scheduler, the semantics
assigns to the TPWN an expected time to termination. Using results
of~\cite{PEVA}, we prove that this expected time is
actually independent of the scheduler, and so that the notion ``expected time
of a TPWN'' is well defined.

We then proceed to study the problem of computing the expected time of a sound
TPWN (loosely speaking, of a TPWN that terminates successfully with probability
1). The expected cost and the expected time have a different interplay with
concurrency. The cost of executing two tasks in parallel is the sum of the
costs (cost models e.g.~salaries of power consumption), while the execution
time of two parallel tasks is the maximum of their individual execution times.
For this reason, standard reward-based algorithms for MDPs, which assume
additivity of the reward along a path, cannot be applied.

Our solution to this problem uses the fact that the expected time of a TPWN is
independent of the scheduler. We define an ``earliest-first'' scheduler which,
loosely speaking, resolves the nondeterminism of the MDP by picking transitions
with earliest possible firing time. Since at first sight the scheduler needs
infinite memory, its corresponding Markov chain is infinite-state, and so of no
help. However, we show how to construct another finite-state Markov chain with
additive rewards, whose expected reward is equal to the expected time of the
infinite-state chain. This finite-state Markov chain can be exponentially
larger than the TPWN, and so our algorithm has exponential complexity. We prove
that computing the expected time is \sharpp{}-hard, even for free-choice TPWNs
in which all transitions times are either $1$ or $0$, and all probabilities are
$1$ or $\nicefrac{1}{2}$. So, in particular, the existence of a polynomial
algorithm implies $\P = \NP$.

In the rest of the paper we show that, despite these negative results, our
algorithm behaves well in practice. For all 642 sound free-choice nets of the
benchmark suite of~\cite{FFJKLVW09}, computing the expected time never takes
longer than a few milliseconds. Looking for a more complicated set of examples,
we study a TPWN computed from a set of logs by process mining. We observe that
the computation of the expected time is sensitive to the distribution of the
execution time of a task. Still, our experiments show that even for
complicated distributions leading to TPWNs with hundreds of transitions and
times spanning two orders of magnitude the expected time can be computed in
minutes.

All missing proofs can be found in the Appendix.

\section{Preliminaries}\label{sec/preliminaries}
We introduce some preliminary definitions. The full version~\cite{FullVersion}
gives more details.

\smallskip

\noindent \textbf{Workflow Nets.} A {\em workflow net} is a tuple $\Net
=(P,T,F,i,o)$ where $P$ and $T$ are disjoint finite sets of \emph{places} and
\emph{transitions}; $F \subseteq (P\times T) \cup (T \times P)$ is a set of
\emph{arcs}; $i, o \in P$ are distinguished {\em initial} and {\em final}
places such that $i$ has no incoming arcs, $o$ has no outgoing arcs, and the
graph $(P \cup T, F \cup \{ (o, i) \})$ is strongly connected. For $x \in P
\cup T$, we write $\preset{x}$ for the set $\{ y \mid (y,x) \in F\}$ and
$\postset{x}$ for $\{ y \mid (x, y) \in F\}$. We call $\preset{x}$ (resp.
$\postset{x}$) the \emph{preset} (resp. \emph{postset}) of $x$. We extend
this notion to sets $X \subseteq P \cup T$ by $\preset{X} \defeq \cup_{x \in X}
\preset{x}$ resp. $\postset{X} \defeq \cup_{x \in X} \postset{x}$. The notions
of marking, enabled transitions, transition firing, firing sequence, and
reachable marking are defined as usual. The {\em initial marking} (resp. {\em
final marking}) of a workflow net, denoted by $\mI$ (resp. $\mO$), has one
token on place $i$ (resp. $o$), and no tokens elsewhere. A firing sequence
$\sigma$ is a \emph{run} if $\mI \trans{\sigma} \mO$, i.e.~if it leads to the
final marking. $\mathit{Run_\Net }$ denotes the set of all runs of $\Net$.

\smallskip

\noindent \textbf{Soundness and 1-safeness.} Well designed workflows should be
free of deadlocks and livelocks. This idea is captured by the notion of
soundness~\cite{DBLP:journals/jcsc/Aalst98,DBLP:journals/fac/AalstHHSVVW11}: A
workflow net is {\em sound} if the final marking is reachable from any
reachable marking.\footnote{In~\cite{DBLP:journals/fac/AalstHHSVVW11}, which
examines many different notions of soundness, this is called {\emph easy
soundness}.} Further, in this paper we restrict ourselves to 1-safe workflows:
A marking $M$ of a workflow net $\W$ is {\em 1-safe} if $M(p) \leq 1 $ for
every place $p$, and $\W$ itself is {\em 1-safe} if every reachable marking is
1-safe. We identify 1-safe markings $M$ with the set $\{ p \in P \mid M(p) = 1 \}$.

\smallskip

\noindent \textbf{Independence, concurrency,
conflict~\cite{DBLP:books/sp/trends86/RozenbergT86}. } Two transitions $t_1$,
$t_2$ of a workflow net are \emph{independent} if $\preset{t_1} \cap
\preset{t_2} = \emptyset$, and \emph{dependent} otherwise. Given a 1-safe
marking $M$, two transitions are {\em concurrent at $M$} if $M$ enables both of
them, and they are independent, and {\em in conflict at $M$} if $M$ enables
both of them, and they are dependent. Finally, we recall the definition of
Mazurkiewicz equivalence. Let $\Net =(P,T,F,i,o)$ be a 1-safe workflow net.
The relation $\equiv_1 \subseteq T^* \times T^*$ is defined as follows: $\sigma
\equiv_1 \tau$ if there are independent transitions $t_1, t_2$ and sequences
$\sigma',\sigma'' \in T^*$ such that $\sigma = \sigma' \, t_1 \, t_2 \sigma''$
and $\tau = \sigma' \, t_2 \, t_1 \sigma''$. Two sequences $\sigma, \tau \in
T^*$ are \emph{Mazurkiewicz equivalent} if $\sigma \equiv \tau$, where $\equiv$
is the reflexive and transitive closure of $\equiv_1$. Observe that $\sigma \in
T^*$ is a firing sequence if{}f every sequence $\tau \equiv \sigma$ is a firing
sequence.

\smallskip

\noindent \textbf{Confusion-freeness, free-choice workflows.} Let $t$ be a
transition of a workflow net, and let $M$ be a 1-safe marking that enables $t$.
The {\em conflict set of $t$ at $M$}, denoted $C(t, M)$, is the set of
transitions in conflict with $t$ at $M$. A set $U$ of transitions is a
\emph{conflict set} of $M$ if there is a transition $t$ such that $U = C(t,M)$.
The conflict sets of $M$ are given by $\mathcal{C}(M) \defeq \cup_{t \in T}
C(t, M)$. A 1-safe workflow net is {\em confusion-free} if for every reachable
marking $M$ and every transition $t$ enabled at $M$, every transition $u$
concurrent with $t$ at $M$ satisfies $C(u, M) = C(u, M \setminus \preset{t}) =
C(u, (M \setminus \preset{t}) \cup \postset{t})$. The following result follows
easily from the definitions (see also~\cite{PEVA}):

\begin{lemma}\cite{PEVA}
\label{lem:confsetspart}
Let $\Net$ be a 1-safe workflow net. If $\Net$ is confusion-free then for every
reachable marking $M$ the conflict sets $\mathcal{C}(M)$ are a partition of the
set of transitions enabled at $M$.
\end{lemma}

A workflow net is \emph{free-choice} if for every two places $p_1,p_2$, if
$\postset{p_1} \cap \postset{p_2} \neq \emptyset$, then $\postset{p_1} =
\postset{p_2}$. Any free-choice net is confusion-free, and the conflict set of
a transition $t$ enabled at a marking $M$ is given by $C(t, M) =
\postset{\left(\preset{t}\right)}$ (see e.g.~\cite{PEVA}).

\section{Timed Probabilistic Workflow Nets}\label{sec/time}
In~\cite{PEVA} we introduced a probabilistic semantics for confusion-free
workflow nets. Intuitively, at every reachable marking a choice between two
concurrent transitions is resolved nondeterministically by a scheduler, while a
choice between two transitions in conflict is resolved probabilistically; the
probability of choosing each transition is proportional to its {\em weight}.
For example, in the net in Fig.~\ref{fig:example2-pwn}, at the marking
$\{p_1,p_3\}$, the scheduler can choose between the conflict sets $\{t_2,t_3\}$
and $\{t_4\}$, and if $\{t_2,t_3\}$ is chosen, then $t_2$ is chosen with
probability $\nicefrac{1}{5}$ and $t_3$ with probability $\nicefrac{4}{5}$. We
extend Probabilistic Workflow Nets by assigning to each transition $t$ a
natural number $\td(t)$ modeling the time it takes for the transition to fire,
once it has been selected.\footnote{The semantics of the model can be defined
in the same way for both discrete and continuous time, but, since our results
only concern discrete time, we only consider this case.}

\begin{definition}[Timed Probabilistic Workflow Nets]
A {\em Timed Probabilistic Workflow Net} \/(TPWN) is a tuple $\W = (\Net ,w,
\td)$ where $\Net =(P,T,F,i,o)$ is a 1-safe confusion-free workflow net, $w
\colon T \ra \Q_{> 0}$ is a \emph{weight function}, and $\td \colon T \ra \N$
is a {\em time function} \/that assigns to every transition a duration.
\end{definition}

\noindent \textbf{Timed sequences.} We assign to each transition sequence
$\sigma$ of $\W$ and each place $p$ a \emph{timestamp} $\mu(\sigma)_p$ through
a \emph{timestamp function} $\mu : T^* \to \Nbot^P$. The set $\Nbot$ is
defined by $\Nbot \defeq \{\bot\} \cup \N$ with $\bot \le x$ and $\bot + x =
\bot$ for all $x \in \Nbot$. Intuitively, if a place $p$ is marked after
$\sigma$, then $\mu(\sigma)_p$ records the ``arrival time'' of the token in
$p$, and if $p$ is unmarked, then $\mu(\sigma)_p = \bot$. When a transition
occurs, it removes all tokens in its preset, and $\tau(t)$ time units later,
puts tokens into its postset. Formally, we define $\mu(\epsilon)_i \defeq 0$,
$\mu(\epsilon)_p \defeq \bot$ for $p \neq i$, and $\mu(\sigma t) \defeq
upd(\mu(\sigma), t)$, where the update function $upd : \Nbot^P \times T \to
\Nbot^P$ is given by:
\begin{align*}
    upd(\vec{x}, t)_p &\defeq
    \begin{cases}
        \max_{q \in \preset{t}} \vec{x}_q + \tau(t) & \text{if }p \in \postset{t} \\
        \bot & \text{if }p \in \preset{t} \setminus \postset{t} \\
        \vec{x}_p & \text{if }p \not\in \preset{t} \cup \postset{t}
    \end{cases}
\end{align*}
We then define $\time(\sigma) \defeq \max_{p \in P} \mu(\sigma)_p$ as the time
needed to fire $\sigma$. Further $\supp{\vec{x}} \defeq \{ p \in P \mid
\vec{x}_p \neq \bot \}$ is the marking represented by a timestamp $\vec{x} \in
\Nbot^P$.

\begin{example}
The net in Fig.~\ref{fig:example2-pwn} is a TPWN. Weights are shown in red next
to transitions, and times are written in blue into the transitions. For the
sequence $\sigma_1 = t_1 t_3 t_4 t_5$, we have $\time(\sigma_1) = 9$, and for
$\sigma_2 = t_1 t_2 t_3 t_4 t_5$, we have $\time(\sigma_2) = 10$. Observe that
the time taken by the sequences is \emph{not} equal to the sum of the durations
of the transitions.
\end{example}

\noindent \textbf{Markov Decision Process semantics.}
A \emph{Markov Decision Process} (MDP) is a tuple $\mathcal{M} =
(Q, q_0, {\it Steps})$ where $Q$ is a finite set of states, $q_0\in Q$ is the
initial state, and ${\it Steps} \colon Q \ra 2^{dist(Q)}$ is the probability
transition function. Paths of an MDP, schedulers, and the probability measure
of paths compatible with a scheduler are defined as usual (see
Appendix~\ref{app-MDP}).

The semantics of a TPWN $\W$ is a Markov Decision Process ${\it MDP}_\W$. The
states of ${\it MDP}_\W$ are either markings $M$ or pairs $(M, t)$, where $t$
is a transition enabled at $M$. The intended meanings of $M$ and $(M, t)$ are
``the current marking is $M$'', and ``the current marking is $M$, and $t$ has
been selected to fire next.'' Intuitively, $t$ is chosen in two steps: first,
a conflict set enabled at $M$ is chosen nondeterministically, and then a
transition of this set is chosen at random, with probability proportional to
its weight.

\begin{figure}
\centering%
\begin{subfigure}[t]{0.49\textwidth}
\centering%
\begin{tikzpicture}[>=latex,scale=0.9,every node/.style={scale=0.9}]

    \node [place,tokens=1,label=left:$i$] at (0,0) (i) {};

    \node [place,label=above:$p_1$] at (-1,-2) (c1) {};
    \node [place,label=below:$p_2$] at (-1,-4) (c2) {};
    \node [place,label=above:$p_3$] at (1,-2) (c3) {};
    \node [place,label=below:$p_4$] at (1,-4) (c4) {};
	\node [place,label=left:$o$] at (0,-6) (o){$o$};

    \node [transition,label=left:\textcolor{tcolor}{$t_1$}] at (0,-1) (t1) {$\ttime{1}$}
	edge [pre] (i)
	edge [post] (c1)
	edge [post] (c3);

    \node [transition,label=above:\textcolor{tcolor}{$t_2$},label={below:\tweight{1}}] at (0,-2) (t2) {$\ttime{4}$};
    \draw[pre] ($(t2.west) + (0,-0.05)$) -- ($(c1.east) + (0,-0.05)$);
    \draw[post] ($(t2.west) + (0,0.05)$) -- ($(c1.east) + (0,0.05)$);

    \node [transition,label=left:\textcolor{tcolor}{$t_3$},label={right:\tweight{4}}] at (-1,-3) (t3) {$\ttime{2}$}
	edge [pre] (c1)
	edge [post] (c2);

    \node [transition,label=left:\textcolor{tcolor}{$t_4$}] at (1,-3) (t4) {$\ttime{5}$}
	edge [pre] (c3)
	edge [post] (c4);

    \node [transition,label=left:\textcolor{tcolor}{$t_5$}] at (0,-5) (t5) {$\ttime{3}$}
	edge [pre] (c2)
	edge [pre] (c4)
	edge [post] (o);

\end{tikzpicture}%
\caption{TPWN $\W$}\label{fig:example2-pwn}%
\end{subfigure}\hfill%
\begin{subfigure}[t]{0.49\textwidth}
\centering%
\begin{tikzpicture}[>=latex,scale=0.7,every node/.style={scale=0.8}]
\tikzstyle{place}=[circle,thick,draw=black,fill=black,minimum size=1mm,inner sep=0mm]
\tikzstyle{transition}=[circle,thick,draw=black,fill=white,minimum size=1mm,inner sep=0mm]

\node [place] at (-2.25,-2.25) (i){};
\node[above=0cm of i]{$\mI$};

\node [place] at (-0.75,-2.25) (it1){};
\node[above=0cm of it1]{\textcolor{tcolor}{$t_1$}};

\node [place] at (0,-3) (p1p3){};
\node[right=0cm of p1p3]{$\{p_1,p_3\}$};

\node [place] at (-1.5,-4.5) (p1p3t2){};
\node[left=0cm of p1p3t2]{\textcolor{tcolor}{$t_2$}};

\node [place] at (0,-4.5) (p1p3t3){};
\node[left=0cm of p1p3t3]{\textcolor{tcolor}{$t_3$}};

\node [place] at (1.5,-4.5) (p1p3t4){};
\node[left=0cm of p1p3t4]{\textcolor{tcolor}{$t_4$}};

\node [place] at (0,-6.0) (p2p3){};
\node[left=0cm of p2p3]{$\{p_2,p_3\}$};

\node [place] at (2.25,-6.0) (p1p4){};
\node[left=0cm of p1p4]{$\{p_1,p_4\}$};

\node [place] at (3.0,-7.5) (p1p4t2){};
\node[left=0cm of p1p4t2]{\textcolor{tcolor}{$t_2$}};

\node [place] at (1.5,-7.5) (p1p4t3){};
\node[left=0cm of p1p4t3]{\textcolor{tcolor}{$t_3$}};

\node [place] at (0,-7.5) (p2p3t4){};
\node[left=0cm of p2p3t4]{\textcolor{tcolor}{$t_4$}};

\node [place] at (0,-9.0) (p2p4){};
\node[left=0cm of p2p4]{$\{p_2,p_4\}$};

\node [place] at (0.75,-9.75) (p2p4t5){};
\node[above=0cm of p2p4t5]{\textcolor{tcolor}{$t_5$}};

\node [place] at (2.25,-9.75) (o){};
\node[above=0cm of o]{$\mO$};

\node [transition] at (-1.5,-2.25) {}
edge [pre,-] (i)
edge [post] (it1);

\node [transition] at (0,-2.25) {}
edge [pre,-] (it1)
edge [post] (p1p3);

\node [transition] at (-0.75,-3.75) {}
edge [pre,-] (p1p3)
edge [post] node[above left=-0.1cm]{$\tweight{\nicefrac{1}{5}}$} (p1p3t2)
edge [post] node[above right=-0.1cm]{$\tweight{\nicefrac{4}{5}}$} (p1p3t3);

\node [transition] at (-1.5,-5.25) {}
edge [pre,-] (p1p3t2)
edge [post, bend left=90, looseness=1.8] (p1p3);

\node [transition] at (0,-5.25) {}
edge [pre,-] (p1p3t3)
edge [post] (p2p3);

\node [transition] at (0.75,-3.75) {}
edge [pre,-] (p1p3)
edge [post] (p1p3t4);

\node [transition] at (2.25,-5.25) {}
edge [pre,-] (p1p3t4)
edge [post] (p1p4);

\node [transition] at (2.25,-6.75) {}
edge [pre,-] (p1p4)
edge [post] node[above right=-0.1cm]{$\tweight{\nicefrac{1}{5}}$} (p1p4t2)
edge [post] node[above left=-0.1cm]{$\tweight{\nicefrac{4}{5}}$} (p1p4t3);

\node [transition] at (1.5,-8.25) {}
edge [pre,-] (p1p4t3)
edge [post] (p2p4);

\node [transition] at (3.0,-8.25) {}
edge [pre,-] (p1p4t2)
edge [post, bend right=90, looseness=1.8] (p1p4);

\node [transition] at (0,-6.75) {}
edge [pre,-] (p2p3)
edge [post] (p2p3t4);

\node [transition] at (0,-8.25) {}
edge [pre,-] (p2p3t4)
edge [post] (p2p4);

\node [transition] at (0,-9.75) {}
edge [pre,-] (p2p4)
edge [post] (p2p4t5);

\node [transition] at (1.5,-9.75) {}
edge [pre,-] (p2p4t5)
edge [post] (o);

\node [transition] at (3.0,-9.75) {}
edge [pre,-,bend right = 45] (o)
edge [post, bend left = 45] (o);

\end{tikzpicture}%
\caption{${\it MDP}_\W$}\label{fig:example2-mdp}%
\end{subfigure}%
\caption{A TPWN and its associated MDP.}\label{fig:example2}%
\end{figure}
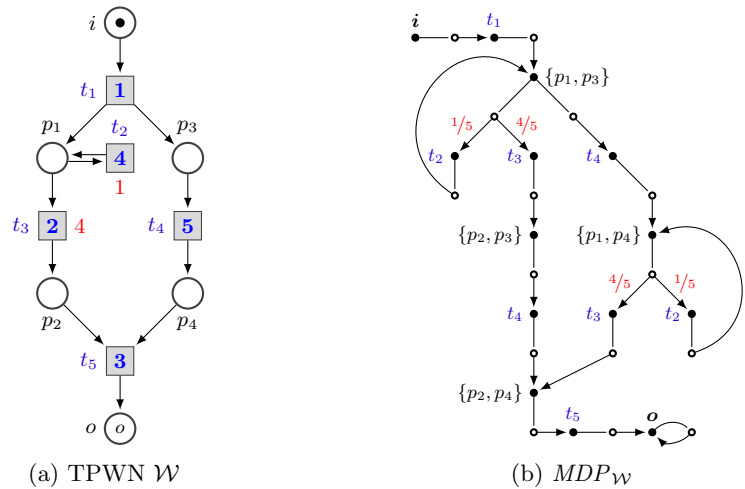

Formally, let $\W=(\Net, w, \tau)$ be a TPWN where $\Net=(P,T,F,i,o)$, let $M$
be a reachable marking of $\W$ enabling at least one transition, and let $C$
be a conflict set of $M$. Let $w(C)$ be the sum of the weights of the
transitions in $C$. The {\em probability distribution $P_{M, C}$ over $T$} is
given by $P_{M, C} (t) = \frac{w(t)}{w(C)}$ if $t \in C$ and $P_{M, C} (t) =0$
otherwise. Now, let ${\cal M}$ be the set of 1-safe markings of $\W$, and let
${\cal E}$ be the set of pairs $(M, t)$ such that $M \in {\cal M}$ and $M$
enables $t$. We define the Markov decision process $\MDP_\W = (Q,q_0,{\it
Steps})$, where $Q= {\cal M} \cup {\cal E}$, $q_0 = \mathbf{i}$, the initial
marking of $\W$, and ${\it Steps}(M)$ is defined for markings of ${\cal M}$ and
${\cal E}$ as follows. For every $M \in {\cal M}$, \begin{itemize} \item if
$M$ enables no transitions, then ${\it Steps}(M)$ contains exactly one
distribution, which assigns probability 1 to $M$, and $0$ to all other states.
\item if $M$ enables at least one transition, then ${\it Steps}(M)$ contains a
distribution $\lambda$ for each conflict set $C$ of $M$. The distribution is
defined by: $\lambda(M, t) = P_{M,C}(t)$ for every $t \in C$, and
$\lambda(s)=0$ for every other state $s$. \end{itemize} For every $(M, t) \in
{\cal E}$, ${\it Steps}(M,t)$ contains one single distribution that assigns
probability 1 to the marking $M'$ such that $M\trans{t}M'$, and probability $0$
to every other state.

\begin{example}
Fig.~\ref{fig:example2-mdp} shows a graphical representation of the MDP of the
TPWN in Fig.~\ref{fig:example2-pwn}. Black nodes represent states, white nodes
probability distributions. A black node $q$ has a white successor for each
probability distribution in ${\it Steps}(q)$. A white node $\lambda$ has a
black successor for each node $q$ such that $\lambda(q) > 0$; the arrow leading
to this black successor is labeled with $\lambda(q)$, unless $\lambda(q)=1$, in
which case there is no label. States $(M, t)$ are abbreviated to $t$.
\end{example}

\noindent \textbf{Schedulers.} Given a TPWN $\W$, a scheduler of ${\it MDP}_W$
is a function $\gamma : T^* \to 2^T$ assigning to each firing sequence $\mI
\trans{\sigma} M$ with $\mathcal{C}(M) \neq \emptyset$ a conflict set
$\gamma(\sigma) \in \mathcal{C}(M)$. A firing sequence $\mI \trans{\sigma} M$
is \emph{compatible} with a scheduler $\gamma$ if for all partitions $\sigma =
\sigma_1 t \sigma_2$ for some transition $t$, we have $t \in \gamma(\sigma_1)$.

\begin{example}
In the TPWN of Fig.~\ref{fig:example2-pwn}, after firing $t_1$ two conflict
sets become concurrently enabled: $\{t_2, t_3\}$ and $\{t_4\}$. A scheduler
picks one of the two. If the scheduler picks $\{t_2, t_3\}$ then $t_2$ may
occur, and in this case, since firing $t_2$ does not change the marking, the
scheduler chooses again one of $\{t_2, t_3\}$ and $\{t_4\}$. So there are
infinitely many possible schedulers, differing only in how many times they pick
$\{t_2, t_3\}$ before picking $t_4$.
\end{example}

\begin{definition}[(Expected) Time until a state is reached]
\label{def:MDPreward}
Let $\pi$ be an infinite path of $\MDP_\W$, and let $M$ be a reachable marking
of $\W$. Observe that $M$ is a state of $\MDP_\W$. The \emph{time needed to
reach $M$ along $\pi$}, denoted $\time(M, \pi)$, is defined as follows: If
$\pi$ does not visit $M$, then $\time(M, \pi) \defeq \infty$; otherwise,
$\time(M, \pi) \defeq \time(\Sigma(\pi'))$, where $\Sigma(\pi')$ is the
transition sequence corresponding to the shortest prefix $\pi'$ of $\pi$ ending
at $M$. Given a scheduler $S$, the expected time until reaching $M$ is defined
as
\[
ET^S_\W(M) \defeq \sum_{\pi \in \mathit{Paths}^S} \time(M,\pi) \cdot \mathit{Prob}^S(\pi).
\]
\noindent and the expected time $ET^S_\W$ is defined as $ET^S_\W \defeq
ET^S_\W(\mO)$, i.e.~the expected time until reaching the final marking.
\end{definition}

In~\cite{PEVA} we proved a result for Probabilistic Workflow Nets (PWNs) with
rewards, showing that the expected reward of a PWN is independent of the
scheduler (intuitively, this is the case because in a confusion-free Petri net
the scheduler only determines the logical order in which transitions occur, but
not which transitions occur). Despite the fact that, contrary to rewards, the
execution time of a firing sequence is not the sum of the execution times of
its transitions, the proof carries over to the expected time with only minor
modifications.

\begin{restatable}{theorem}{thmexpectedtime}\label{thm:expected-time}
Let $\W$ be a TPWN.
\begin{itemize}
\item[(1)] There exists a value $ET_\W$ such that for every scheduler $S$ of
    $\W$, the expected time $ET^S_\W$ of $\W$ under $S$ is equal to $ET_\W$.
\item[(2)] $ET_\W$ is finite if{}f $\W$ is sound.
\end{itemize}
\end{restatable}

By this theorem, the expected time $ET_\W$ can be computed by choosing a
suitable scheduler $S$, and computing $ET^S_\W$.

\section{Computation of the expected time}\label{sec:computation}
We show how to compute the expected time of a TPWN\@. We fix an appropriate
scheduler, show that it induces a finite-state Markov chain, define an
appropriate reward function for the chain, and prove that the expected time is
equal to the expected reward.

\subsection{Earliest-First Scheduler}\label{subsec:scheduler}

Consider a firing sequence $\mI \trans{\sigma} M$. We define the \emph{starting
time} of a conflict set $C \in \mathcal{C}(M)$ as the earliest time at which
the transitions of $C$ become enabled. This occurs after \emph{all} tokens of
$\preset{C}$ arrive\footnote{This is proved in Lemma~\ref{lem:start-time} in the
Appendix.}, and so the starting time of $C$ is the
maximum of $\mu(\sigma)_p$ for $p \in \preset{C}$ (recall that $\mu(\sigma)_p$
is the latest time at which a token arrives at $p$ while firing $\sigma$).

Intuitively, the ``earliest-first'' scheduler always chooses the conflict set
with the earliest starting time (if there are multiple such conflict sets, the
scheduler chooses any one of them). Formally, recall that a scheduler is a
mapping $\gamma \colon T^* \rightarrow 2^T$ such that for every firing sequence
$\mI \trans{\sigma} M$, the set $\gamma(\sigma)$ is a conflict set of $M$. We
define the \emph{earliest-first scheduler} $\gamma$ by:
\begin{align*}
    \gamma(\sigma) &\defeq
    \argmin_{C \in \mathcal{C}(M)} \, \max_{p \in \preset{C}} \, \mu(\sigma)_p &
        \text{where $M$ is given by $\mI \trans{\sigma} M$.}
\end{align*}

\begin{example}
Fig.~\ref{fig:scheduler-mc-infinite} shows the Markov chain induced by the
``earliest-first'' scheduler defined above in the MDP of
Fig.~\ref{fig:example2-mdp}. Initially we have a token at $\mI$ with arrival
time $0$. After firing $t_1$, which takes time $1$, we obtain tokens in $p_1$
and $p_3$ with arrival time $1$. In particular, the conflict sets $\{t_2,t_3\}$
and $\{t_4\}$ become enabled at time $1$. The scheduler can choose any of them,
because they have the same starting time. Assume it chooses $\{t_2,t_3\}$. The
Markov chain now branches into two transitions, corresponding to firing $t_2$
and $t_3$ with probabilities $\nicefrac{1}{5}$ and $\nicefrac{4}{5}$,
respectively. Consider the branch in which $t_2$ fires. Since $t_2$ starts at
time $1$ and takes $4$ time units, it removes the token from $p_1$ at time $1$,
and adds a new token to $p_1$ with arrival time $5$; the token at $p_3$ is not
affected, and it keeps its arrival time of $1$. So we have $\mu(t_1t_2) =
\left\{ \sstack{p_1}{5},\sstack{p_3}{1} \right\}$ (meaning
$\mu(t_1t_2)_{p_1}=5$, $\mu(t_1t_2)_{p_3}=1$, and $\mu(t_1t_2)_{p}= \bot$
otherwise). Now the conflict sets $\{t_2,t_3\}$ and $\{t_4\}$ are enabled
again, but with a difference: while $\{t_4\}$ has been enabled since time $1$,
the set $\{t_2,t_3\}$ is now enabled since time $\mu(t_1 t_2)_{p_1} = 5$. The
scheduler must now choose $\{t_4\}$, leading to the marking that puts tokens on
$p_1$ and $p_4$ with arrival times $\mu(t_1 t_2 t_4)_{p_1} = 5$ and $\mu(t_1
t_2 t_4)_{p_4} = 6$. In the next steps the scheduler always chooses
$\{t_2,t_3\}$ until $t_5$ becomes enabled. The final marking $\mO$ can be
reached after time $9$, through $t_1 t_3 t_4 t_5$ with probability
$\nicefrac{4}{5}$, or with times $10 + 4k$ for $k \in \N$, through $t_1 t_2 t_4
t_2^k t_3 t_5$ with probability $\left(\nicefrac{1}{5}\right)^{k+1} \cdot
\nicefrac{4}{5}$ (the times at which the final marking can be reached are
written in blue inside the final states).
\end{example}

Theorem~\ref{thm:finitememory} below shows that the earliest-first scheduler
only needs finite memory, which is not clear from the definition. The
construction is similar to those of~\cite{GaubertM99a,GaubertM99b,CarlierC87}.
However, our proof crucially depends on TPWNs being confusion-free.

\begin{restatable}{theorem}{thmfinitememory}\label{thm:finitememory}
Let $H \defeq \max_{t \in T} \td(t)$ be the maximum duration of the transitions
of $T$, and let $\Hbot \defeq \{\bot, 0,1,\ldots, H\} \subseteq \Nbot$. There
are functions $\nu \colon T^* \rightarrow \Hbot^P$ (compare with $\mu \colon
T^* \rightarrow \Nbot^P$), $f \colon \Hbot^P \times T \rightarrow \Hbot^P$ and
$r \colon \Hbot^P \to \N$ such that for every $\sigma = t_1 \ldots t_n \in T^*$
compatible with $\gamma$ and for every $t \in T$ enabled by $\sigma$:
\begin{align}
    \gamma(\sigma) & = \argmin_{C \in \mathcal{C}(\supp{\nu(\sigma)})} \,
            \max_{p \in \preset{C}} \, \nu(\sigma)_p \label{prop1} \\
    \nu(\sigma t) & = f(\nu(\sigma), t) \label{prop2} \\
    \time(\sigma) &= \max_{p \in P} \nu(\sigma)_p +
            \sum_{k=0}^{n-1} r(\nu(t_1 \ldots t_k)) \label{prop3}
\end{align}
\end{restatable}

Observe that, unlike $\mu$, the range of $\nu$ is finite. We call it the
\emph{finite abstraction} of $\mu$. Equation~\ref{prop1} states that $\gamma$
can be computed directly from the finite abstraction $\nu$.
Equation~\ref{prop2} shows that $\nu(\sigma t)$ can be computed from
$\nu(\sigma)$ and $t$. So $\gamma$ only needs to remember an element of
$\Hbot^P$, which implies that it only requires finite memory. Finally, observe
that the function $r$ of Equation~\ref{prop3} has a finite domain, and so it
allows us to use $\nu$ to compute the time needed by $\sigma$.

\begin{figure}
\centering%
\begin{subfigure}[t]{0.49\textwidth}
\centering%
\begin{tikzpicture}[>=latex,scale=0.8,every node/.style={scale=0.8}]
\tikzstyle{place}=[circle,thick,draw=black,fill=black!10!white,minimum size=6mm,inner sep=0.5mm]
\tikzstyle{transition}=[circle,thick,draw=black,fill=black,minimum size=1mm,inner sep=0mm]

\node [transition] at (-2.25,-2.25) (i){};
\node[left=0cm of i]{$\{\sstack{i}{0}\}$};

\node [transition] at (-0.75,-2.25) (it1){};
\node[above=0cm of it1]{\textcolor{tcolor}{$t_1$}};

\node [transition] at (0,-3) (p1p3u11){};
\node[below=0cm of p1p3u11,yshift=0em]{$\{\sstack{p_1}{1},\sstack{p_3}{1}\}$};

\node [transition] at (-1.5,-3.75) (p1p3u11t2){};
\node[left=0cm of p1p3u11t2]{\textcolor{tcolor}{$t_2$}};

\node [transition] at (1.5,-3.75) (p1p3u11t3){};
\node[right=0cm of p1p3u11t3]{\textcolor{tcolor}{$t_3$}};

\node [transition] at (-1.5,-4.5) (p1p3u51){};
\node[left=0cm of p1p3u51,yshift=0em]{$\{\sstack{p_1}{5},\sstack{p_3}{1}\}$};

\node [transition] at (-1.5,-5.25) (p1p3u51t4){};
\node[left=0cm of p1p3u51t4]{\textcolor{tcolor}{$t_4$}};

\node [transition] at (-1.5,-6) (p1p4u56){};
\node[left=0cm of p1p4u56,yshift=0.25em]{$\{\sstack{p_1}{5},\sstack{p_4}{6}\}$};

\node [transition] at (-1.5,-6.75) (p1p4u56t2){};
\node[left=0cm of p1p4u56t2]{\textcolor{tcolor}{$t_2$}};

\node [transition] at (-0.75,-6) (p1p4u56t3){};
\node[above=0cm of p1p4u56t3]{\textcolor{tcolor}{$t_3$}};

\node [transition] at (0,-6) (p2p4u76){};
\node[above=0cm of p2p4u76,xshift=0.5em]{$\{\sstack{p_2}{7},\sstack{p_4}{6}\}$};

\node [transition] at (0,-6.75) (p2p4u76t5){};
\node[left=0cm of p2p4u76t5]{\textcolor{tcolor}{$t_5$}};

\node [transition] at (1.5,-4.5) (p2p3u31){};
\node[right=0cm of p2p3u31,yshift=0em]{$\{\sstack{p_2}{3},\sstack{p_3}{1}\}$};

\node [transition] at (1.5,-5.25) (p2p3u31t4){};
\node[right=0cm of p2p3u31t4]{\textcolor{tcolor}{$t_4$}};

\node [transition] at (1.5,-6) (p2p4u36){};
\node[right=0cm of p2p4u36,yshift=0em]{$\{\sstack{p_2}{3},\sstack{p_4}{6}\}$};

\node [transition] at (1.5,-6.75) (p2p4u36t5){};
\node[right=0cm of p2p4u36t5]{\textcolor{tcolor}{$t_5$}};

\node [transition] at (-1.5,-7.5) (p1p4u96){};
\node[left=0cm of p1p4u96,yshift=0.5em]{$\{\sstack{p_1}{9},\sstack{p_4}{6}\}$};

\node [transition] at (-2.5,-8.25) (p1p4u96t2){};
\node[left=0cm of p1p4u96t2]{\textcolor{tcolor}{$t_2$}};

\node [transition] at (-0.5,-8.25) (p1p4u96t3){};
\node[right=0cm of p1p4u96t3]{\textcolor{tcolor}{$t_3$}};

\node [transition] at (-2.5,-9) (p1p4u136){};
\node[right=0cm of p1p4u136,yshift=0.5em]{$\{\sstack{p_1}{13},\sstack{p_4}{6}\}$};

\node [transition] at (-0.5,-9) (p2p4u116){};
\node[below=0cm of p2p4u116,xshift=0.5em]{$\{\sstack{p_2}{11},\sstack{p_4}{6}\}$};

\node [transition] at (0.5,-9) (p2p4u116t5){};
\node[above=0cm of p2p4u116t5]{\textcolor{tcolor}{$t_5$}};

\node [place] at (1.5,-7.5) (ou9){\sreward{9}};
\node[right=0cm of ou9]{$\{\sstack{o}{9}\}$};

\node [place] at (0,-7.5) (ou10){\sreward{10}};
\node[right=0cm of ou10]{$\{\sstack{o}{10}\}$};

\node [place] at (1.5,-9) (ou14){\sreward{14}};
\node[right=0cm of ou14]{$\{\sstack{o}{14}\}$};

\node at (-3.0,-9.75) (p1p4u136t2){$\dots$};
\node[below=0cm of p1p4u136t2]{$\phantom{t_2}$};

\node at (-1.75,-9.75) (p1p4u136t3){$\dots$};
\node[below=0cm of p1p4u136t3]{$\phantom{t_3}$};

\draw[->] (i) -- (it1);
\draw[->] (it1) -| (p1p3u11);

\draw[->] (p1p3u11) -|  node[above=0.0cm,xshift=2em]{$\tweight{\nicefrac{1}{5}}$} (p1p3u11t2);
\draw[->] (p1p3u11) -|  node[above=0.0cm,xshift=-2em]{$\tweight{\nicefrac{4}{5}}$} (p1p3u11t3);

\draw[->] (p1p3u11t2) -| (p1p3u51);
\draw[->] (p1p3u11t3) -| (p2p3u31);

\draw[->] (p1p3u51) -- (p1p3u51t4);
\draw[->] (p1p3u51t4) -- (p1p4u56);

\draw[->] (p1p4u56) --  node[left=0.0cm]{$\tweight{\nicefrac{1}{5}}$} (p1p4u56t2);
\draw[->] (p1p4u56) --  node[below=0.0cm]{$\tweight{\nicefrac{4}{5}}$} (p1p4u56t3);
\draw[->] (p1p4u56t2) -- (p1p4u96);
\draw[->] (p1p4u56t3) -- (p2p4u76);

\draw[->] (p2p4u76) -- (p2p4u76t5);
\draw[->] (p2p4u76t5) -- (ou10);

\draw[->] (p2p3u31) -- (p2p3u31t4);
\draw[->] (p2p3u31t4) -- (p2p4u36);

\draw[->] (p2p4u36) -- (p2p4u36t5);
\draw[->] (p2p4u36t5) -- (ou9);

\draw[->] (p1p4u96) --  node[above left=-0.15cm and -0.05cm]{$\tweight{\nicefrac{1}{5}}$} (p1p4u96t2);
\draw[->] (p1p4u96) --  node[above right=-0.15cm and -0.05cm]{$\tweight{\nicefrac{4}{5}}$} (p1p4u96t3);
\draw[->] (p1p4u96t2) -- (p1p4u136);
\draw[->] (p1p4u96t3) -- (p2p4u116);

\draw[->] (p1p4u136) --  node[left=0.05cm]{$\tweight{\nicefrac{1}{5}}$} (p1p4u136t2);
\draw[->] (p1p4u136) --  node[right=0.05cm]{$\tweight{\nicefrac{4}{5}}$} (p1p4u136t3);

\draw[->] (p2p4u116) -- (p2p4u116t5);
\draw[->] (p2p4u116t5) -- (ou14);

\draw[->] (ou10) edge[loop below] (ou9);
\draw[->] (ou9) edge[loop below] (ou10);
\draw[->] (ou14) edge[loop below] (ou14);

\end{tikzpicture}
\caption{Infinite MC for scheduler using $\mu(\sigma)$, with final states labeled by $\time(\sigma)$.}\label{fig:scheduler-mc-infinite}%
\end{subfigure}\hfill%
\begin{subfigure}[t]{0.49\textwidth}
\centering%
\begin{tikzpicture}[>=latex,scale=0.8,every node/.style={scale=0.8}]
\tikzstyle{place}=[circle,thick,draw=black,fill=black!10!white,minimum size=1mm,inner sep=0.5mm]
\tikzstyle{transition}=[circle,thick,draw=black,fill=black,minimum size=1mm,inner sep=0mm]

\node [place] at (-2.25,-2.25) (i){\sreward{0}};
\node[left=0cm of i]{$\{\sstack{i}{0}\}$};

\node [transition] at (-0.75,-2.25) (it1){};
\node[above=0cm of it1]{\textcolor{tcolor}{$t_1$}};

\node [place] at (0,-3) (p1p3u11){\sreward{1}};
\node[below=0cm of p1p3u11,yshift=0em]{$\{\sstack{p_1}{1},\sstack{p_3}{1}\}$};

\node [transition] at (-1.5,-3.75) (p1p3u11t2){};
\node[left=0cm of p1p3u11t2]{\textcolor{tcolor}{$t_2$}};

\node [transition] at (1.5,-3.75) (p1p3u11t3){};
\node[right=0cm of p1p3u11t3]{\textcolor{tcolor}{$t_3$}};

\node [place] at (-1.5,-4.5) (p1p3u51){\sreward{0}};
\node[left=0cm of p1p3u51,yshift=0em]{$\{\sstack{p_1}{4},\sstack{p_3}{0}\}$};

\node [transition] at (-1.5,-5.25) (p1p3u51t4){};
\node[left=0cm of p1p3u51t4]{\textcolor{tcolor}{$t_4$}};

\node [place] at (-1.5,-6) (p1p4u56){\sreward{4}};
\node[left=0cm of p1p4u56,yshift=0.25em]{$\{\sstack{p_1}{4},\sstack{p_4}{5}\}$};

\node [transition] at (-1.5,-6.75) (p1p4u56t2){};
\node[left=0cm of p1p4u56t2]{\textcolor{tcolor}{$t_2$}};

\node [transition] at (-0.75,-6) (p1p4u56t3){};
\node[above=0cm of p1p4u56t3]{\textcolor{tcolor}{$t_3$}};

\node [place] at (0,-6) (p2p4u76){\sreward{2}};
\node[above=0cm of p2p4u76,xshift=0.5em]{$\{\sstack{p_2}{2},\sstack{p_4}{1}\}$};

\node [transition] at (0,-6.75) (p2p4u76t5){};
\node[left=0cm of p2p4u76t5]{\textcolor{tcolor}{$t_5$}};

\node [place] at (1.5,-4.5) (p2p3u31){\sreward{0}};
\node[right=0cm of p2p3u31,yshift=0em]{$\{\sstack{p_2}{2},\sstack{p_3}{0}\}$};

\node [transition] at (1.5,-5.25) (p2p3u31t4){};
\node[right=0cm of p2p3u31t4]{\textcolor{tcolor}{$t_4$}};

\node [place] at (1.5,-6) (p2p4u36){\sreward{5}};
\node[right=0cm of p2p4u36,yshift=0em]{$\{\sstack{p_2}{2},\sstack{p_4}{5}\}$};

\node [transition] at (1.5,-6.75) (p2p4u36t5){};
\node[right=0cm of p2p4u36t5]{\textcolor{tcolor}{$t_5$}};

\node [place] at (-1.5,-7.5) (p1p4u96){\sreward{4}};
\node[left=0cm of p1p4u96,yshift=0.5em]{$\{\sstack{p_1}{4},\sstack{p_4}{1}\}$};

\node [transition] at (-2.5,-8.25) (p1p4u96t2){};
\node[left=0cm of p1p4u96t2]{\textcolor{tcolor}{$t_2$}};

\node [transition] at (-0.5,-8.25) (p1p4u96t3){};
\node[right=0cm of p1p4u96t3]{\textcolor{tcolor}{$t_3$}};

\node [place] at (-2.5,-9) (p1p4u136){\sreward{4}};
\node[right=0cm of p1p4u136,yshift=0.75em]{$\{\sstack{p_1}{4},\sstack{p_4}{0}\}$};

\node [place] at (-0.5,-9) (p2p4u126){\sreward{2}};
\node[below=0cm of p2p4u126,xshift=1.0em]{$\{\sstack{p_2}{2},\sstack{p_4}{0}\}$};

\node [transition] at (0.5,-9) (p2p4u126t5){};
\node[above=0cm of p2p4u126t5]{\textcolor{tcolor}{$t_5$}};

\node [place] at (1.5,-9) (ou15){\sreward{3}};
\node[right=0cm of ou15]{$\{\sstack{o}{3}\}$};

\node [transition] at (-2.5,-9.75) (p1p4u136t2){};
\node[below=0cm of p1p4u136t2]{\textcolor{tcolor}{$t_2$}};

\node [transition] at (-1.5,-9.75) (p1p4u136t3){};
\node[below=0cm of p1p4u136t3]{\textcolor{tcolor}{$t_3$}};

\draw[->] (i) -- (it1);
\draw[->] (it1) -| (p1p3u11);

\draw[->] (p1p3u11) -|  node[above=0.0cm,xshift=2em]{$\tweight{\nicefrac{1}{5}}$} (p1p3u11t2);
\draw[->] (p1p3u11) -|  node[above=0.0cm,xshift=-2em]{$\tweight{\nicefrac{4}{5}}$} (p1p3u11t3);

\draw[->] (p1p3u11t2) -| (p1p3u51);
\draw[->] (p1p3u11t3) -| (p2p3u31);

\draw[->] (p1p3u51) -- (p1p3u51t4);
\draw[->] (p1p3u51t4) -- (p1p4u56);

\draw[->] (p1p4u56) --  node[left=0.0cm]{$\tweight{\nicefrac{1}{5}}$} (p1p4u56t2);
\draw[->] (p1p4u56) --  node[below=0.0cm]{$\tweight{\nicefrac{4}{5}}$} (p1p4u56t3);
\draw[->] (p1p4u56t2) -- (p1p4u96);
\draw[->] (p1p4u56t3) -- (p2p4u76);

\draw[->] (p2p4u76) -- (p2p4u76t5);
\draw[->] (p2p4u76t5) -- (ou15);

\draw[->] (p2p3u31) -- (p2p3u31t4);
\draw[->] (p2p3u31t4) -- (p2p4u36);

\draw[->] (p2p4u36) -- (p2p4u36t5);
\draw[->] (p2p4u36t5) -- (ou15);

\draw[->] (p1p4u96) --  node[above left=-0.15cm and -0.05cm]{$\tweight{\nicefrac{1}{5}}$} (p1p4u96t2);
\draw[->] (p1p4u96) --  node[above right=-0.15cm and -0.05cm]{$\tweight{\nicefrac{4}{5}}$} (p1p4u96t3);
\draw[->] (p1p4u96t2) -- (p1p4u136);
\draw[->] (p1p4u96t3) -- (p2p4u126);

\draw[->] (p1p4u136) --  node[left=0.0]{$\tweight{\nicefrac{1}{5}}$} (p1p4u136t2);
\draw[->] (p1p4u136) --  node[above right=-0.15cm]{$\tweight{\nicefrac{4}{5}}$} (p1p4u136t3);
\draw[->] (p1p4u136t2) edge[post, bend left=100, looseness=2.0] (p1p4u136);
\draw[->] (p1p4u136t3) -- (p2p4u126);

\draw[->] (p2p4u126) -- (p2p4u126t5);
\draw[->] (p2p4u126t5) -- (ou15);

\draw[->] (ou15) edge[loop below] (ou15);

\end{tikzpicture}
\caption{Finite MC for scheduler using $\nu(\sigma)$, with states labeled by rewards $r(\nu(\sigma))$.}\label{fig:scheduler-mc-finite}%
\end{subfigure}%
\caption{Two Markov chains for the ``earliest-first'' scheduler.}\label{fig:scheduler-mc}%
\end{figure}

The formal definition of the functions $\nu$, $f$, and $r$ is given below,
together with the definition of the auxiliary operator $\ominus \colon \Nbot^P
\times \N \to \Nbot^P$:
\begin{align*}
    \left(\vec{x} \ominus n\right)_p &\defeq
    \begin{cases}
        \max(\vec{x}_p - n, 0) & \text{if }\vec{x}_p \neq \bot \\
        \bot & \text{if }\vec{x}_p = \bot \\
    \end{cases} &
    f(\vec{x}, t) &\defeq upd(\vec{x}, t) \ominus \max_{p \in \preset{t}} \vec{x}_p \\
    \nu(\epsilon) &\defeq \mu(\epsilon) \text{ and }
    \nu(\sigma t) \defeq \mu(\sigma t) \ominus \max_{p \in \preset{t}} \mu(\sigma)_p  &
    r(\vec{x}) &\defeq \min_{C \in \mathcal{C}(\supp{\vec{x}})}\max_{p \in \preset{C}} \vec{x}_p
\end{align*}

\begin{example}
Fig.~\ref{fig:scheduler-mc-finite} shows the finite-state Markov chain induced
by the ``earliest-first'' scheduler computed using the abstraction $\nu$.
Consider the firing sequence $t_1t_3$. We have $\mu(t_1t_3) = \left\{
\sstack{p_2}{3},\sstack{p_3}{1} \right\}$, i.e.~the tokens in $p_2$ and $p_3$
arrive at times $3$ and $1$, respectively. Now we compute $\nu(t_1t_3)$, which
corresponds to the \emph{local} arrival times of the tokens, i.e.~the time
elapsed \emph{since the last transition starts to fire until the token
arrives}. Transition $t_3$ starts to fire at time $1$, and so the local arrival
times of the tokens in $p_2$ and $p_3$ are $2$ and $0$, respectively, i.e.~we
have $\nu(t_1t_3) = \left\{ \sstack{p_2}{2},\sstack{p_3}{0} \right\}$. Using
these local times we compute the local starting time of the conflict sets
enabled at $\{p_2, p_3\}$. The scheduler always chooses the conflict set with
earliest local starting time. In Fig.~\ref{fig:scheduler-mc-finite} the
earliest local starting time of the state reached by firing $\sigma$, which is
denoted $r(\nu(\sigma))$, is written in blue inside the state. The theorem
above shows that this scheduler always chooses the same conflict sets as the
one which uses the function $\mu$, and that the time of a sequence can be
obtained by adding the local starting times. This allows us to consider the
earliest local starting time of a state as a \emph{reward} associated to the
state; then, the time taken by a sequence is equal to the sum of the rewards
along the corresponding path of the chain. For example, we have $\time(t_1 t_2
t_4 t_3 t_5) = 0 + 1 + 0 + 4 + 2 + 3 = 10$.

Finally, let us see how $\nu(\sigma t)$ is computed from $\nu(\sigma)$ for
$\sigma = t_1 t_2 t_4$ and $t=t_2$. We have $\nu(\sigma) = \left\{
\sstack{p_1}{4},\sstack{p_4}{5} \right\}$, i.e.~the local arrival times for
the tokens in $p_1$ and $p_4$ are $4$ and $5$, respectively. Now $\{t_2,t_3\}$
is scheduled next, with local starting time $r(\nu(\sigma)) = \nu(\sigma)_{p_1}
= 4$. If $t_2$ fires, then, since $\tau(t_2)=4$, we first add $4$ to the time
of $p_1$, obtaining $\left\{ \sstack{p_1}{8},\sstack{p_4}{5} \right\}$. Second,
we subtract $4$ from \emph{all} times, to obtain the time elapsed since $t_2$
started to fire (for local times the origin of time changes every time a
transition fires), yielding the final result $\nu(\sigma t_2)=\left\{
\sstack{p_1}{4},\sstack{p_4}{1} \right\}$.
\end{example}

\subsection{Computation in the probabilistic case}\label{subsec:computation-probabilistic}

Given a TPWN and its corresponding MDP, in the previous section we have defined
a finite-state earliest-first scheduler and a reward function of its induced
Markov chain. The reward function has the following property: the execution
time of a firing sequence compatible with the scheduler is equal to the sum of
the rewards of the states visited along it. From the theory of Markov chains
with rewards, it follows that the expected accumulated reward until reaching a
certain state, provided that this state is reached with probability $1$, can be
computed by solving a linear equation system. We use this result to compute
the expected time $ET_W$.

Let $\W$ be a sound TPWN. For every firing sequence $\sigma$ compatible with
the earliest-first scheduler $\gamma$, the finite-state Markov chain induced by
$\gamma$ contains a state $\vec{x} = \nu(\sigma) \in \Hbot^P$. Let
$C_{\vec{x}}$ be the conflict set scheduled by $\gamma$ at $\vec{x}$.  We
define a system of linear equations with variables $X_{\vec{x}}$, one for each
state $\vec{x}$:
\begin{equation}
    \begin{aligned}
        X_{\vec{x}} &= r(\vec{x}) + \sum_{t \in C_{\vec{x}}}
                \frac{w(t)}{w(C_{\vec{x}})} \cdot X_{f(\vec{x}, t )}
        & \text{ if }\supp{\vec{x}} \neq \mO \\
        X_{\vec{x}} &= \max_{p \in P} \vec{x}_p
        & \text{if }\supp{\vec{x}} = \mO
    \end{aligned}
\label{eqn:lineq}
\end{equation}
The solution of the system is the expected reward of a path leading from $\mI$
to $\mO$. By the theory of Markov chains with rewards/costs (\cite{Baier2008},
Chapter 10.5), we have:

\begin{lemma}
Let $\W$ be a sound TPWN. Then the system of linear equations~(\ref{eqn:lineq})
has a unique solution $\vec{X}$, and $ET_\W = \vec{X}_{\nu(\epsilon)}$.
\label{lem:exp-solution-lineq}
\end{lemma}

\begin{theorem}
    Let $\W$ be a TPWN. Then $ET_\W$ is either $\infty$ or a rational number
    and can be computed in single exponential time.
\end{theorem}
\begin{proof}
    We assume that the input has size $n$ and all times and weights are given
    in binary notation. Testing whether $\W$ is sound can be done by
    exploration of the state space of reachable markings in time $\lO(2^n)$. If
    $\W$ is unsound, we have $ET_\W = \infty$.

    Now assume that $\W$ is sound. By Lemma~\ref{lem:exp-solution-lineq},
    $ET_\W$ is the solution to the linear equation system~(\ref{eqn:lineq}),
    which is finite and has rational coefficients, so it is a rational number.
    The number of variables $\abs{\vec{X}}$ of~(\ref{eqn:lineq}) is bounded by
    the size of $\Hbot^P$, and as $H = \max_{t \in T} \td(t)$ we have
    $\abs{\vec{X}} \le (1 + H)^{\abs{P}} \le \left(1 + 2^n\right)^n \le 2^{n^2
    + n}$. The linear equation system can be solved in time $\lO\left(n^2
    \cdot \abs{\vec{X}}^3\right)$ and therefore in time $\lO(2^{p(n)})$ for some
    polynomial $p$.
\end{proof}

\section{Lower bounds for the expected time}\label{sec:lowerbound}
We analyze the complexity of computing the expected time of a TPWN\@. Botezano
{\it et al.} show in~\cite{BVT16} that deciding if the expected time exceeds a
given bound is \NP-hard. However, their reduction produces TPWNs with weights
and times of arbitrary size. An open question is if the expected time can be
computed in polynomial time when the times (and weights) must be taken from a
finite set. We prove that this is not the case unless $\P = \NP$, even if all
times are $0$ or $1$, all weights are $1$, the workflow net is sound, acyclic
and free-choice, and the size of each conflict set is at most $2$ (resulting
only in probabilities $1$ or $\nicefrac{1}{2}$). Further, we show that even
computing an $\epsilon$-approximation is equally hard. These two results above
are a consequence of the main theorem of this section: computing the expected
time is \sharpp-hard~\cite{Valiant79}. For example, counting the number of
satisfying assignments for a boolean formula (\#SAT) is a \sharpp-complete
problem. Therefore a polynomial-time algorithm for a \sharpp-hard problem
would imply $\P = \NP$.

The problem used for the reduction is defined on PERT
networks~\cite{Elmaghraby77}, in the specialized form of \emph{two-state
stochastic PERT networks}~\cite{Hagstrom88}, described below.

\begin{definition}
    A \emph{two-state stochastic PERT network} is a tuple $\PN = (G, s, t,
    \vec{p})$, where $G = (V,E)$ is a directed acyclic graph with vertices $V$,
    representing events, and edges $E$, representing tasks, with a single
    source vertex $s$ and sink vertex $t$, and where the vector $\vec{p} \in
    \Q^E$ assigns to each edge $e \in E$ a rational probability $p_e \in
    [0,1]$. We assume that all $p_e$ are written in binary.

    Each edge $e \in E$ of $\PN$ defines a random variable $X_e$ with
    distribution $\Pr(X_e = 1) = p_e$ and $\Pr(X_e = 0) = 1-p_e$. All $X_e$ are
    assumed to be independent. The \emph{project duration} $PD$ of $\PN$ is the
    length of the longest path in the network
    \[
        PD(\PN) \defeq \max_{\pi \in \Pi} \sum_{e \in \pi} X_e
    \]
    where $\Pi$ is the set of paths from vertex $s$ to vertex $t$. As this
    defines a random variable, the \emph{expected project duration} of $\PN$ is
    then given by $\Ex(PD(\PN))$.
\end{definition}

\begin{example}
Figure~\ref{fig:reduction-pert} shows a small PERT network (without $\vec{p}$),
where the project duration depends on the paths
$\Pi = \set{ e_1 e_3 e_6, e_1 e_4 e_7, e_2 e_5 e_7}$.
\end{example}

The following problem is \sharpp-hard (from~\cite{Hagstrom88}, using the
results from~\cite{ProvanB83}):

\begin{quote}
{\bf Given}: A two-state stochastic PERT network $\PN$. \\
{\bf Compute}: The expected project duration $\Ex(PD(\PN))$.
\end{quote}

\subsubsection{First reduction: 0/1 times, arbitrary weights.}
We reduce the problem above to computing the expected time of an acyclic TPWN
with 0/1 times but arbitrary weights. Given a two-state stochastic PERT
network $\PN$, we construct a timed probabilistic workflow net $\W_\PN$ as
follows:

\begin{itemize}
    \item For each edge $e = (u,v) \in E$, add the ``gadget net'' shown in
        Figure~\ref{fig:reduction-weight-rational}.  Assign $w(t_{e,0}) =
        1-p_e$, $w(t_{e,1}) = p_e$, $\tau(t_{e,0}) = 0$, and $\tau(t_{e,1}) =
        1$.
    \item For each vertex $v \in V$, add a transition $t_v$ with arcs from each
        $[e, v]$ such that $e = (u,v) \in E$ for some $u$ and arcs to each $[v,
        e]$ such that $e = (v,w) \in E$ for some $w$.  Assign $w(t_v) = 1$ and
        $\tau(t_v) = 0$.
    \item Add the place $i$ with an arc to $t_s$ and the place $o$ with an arc
        from $t_t$.
\end{itemize}

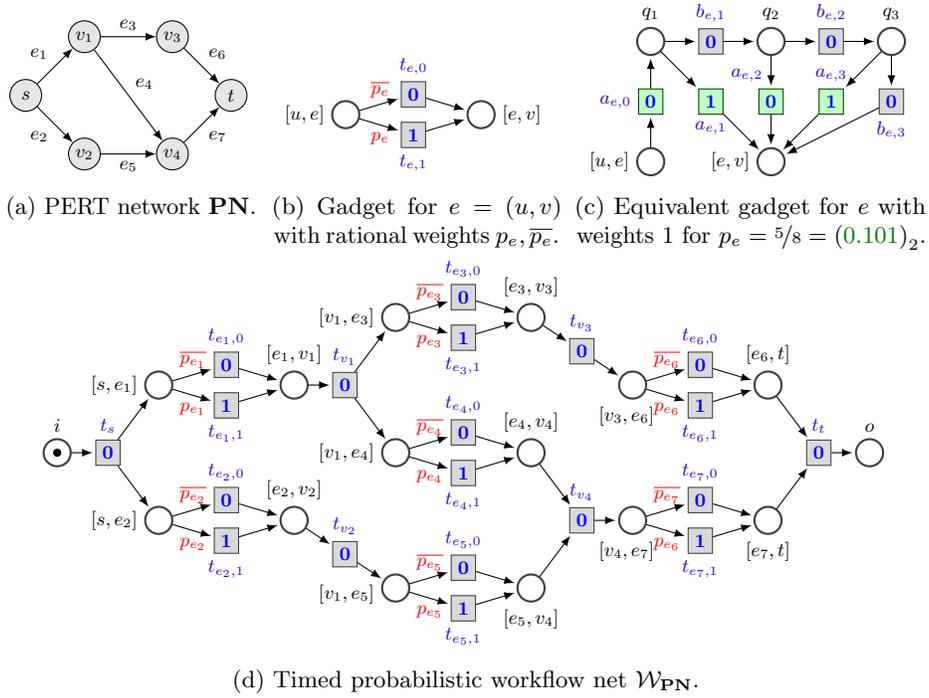
\begin{figure}[t]
\centering
\begin{subfigure}[t]{0.27\textwidth}
    \centering
    \begin{tikzpicture}[>=latex,scale=0.78,every node/.style={scale=0.8}]
    \tikzstyle{vertex}=[circle, draw, fill=black!10, inner sep=0pt, minimum size=15pt]
    \tikzstyle{weight} = [font=\small]
    \tikzstyle{edge} = [draw,thick,-]

    \node[vertex] (s)   at (0,1)   {$s$};
    \node[vertex] (v_1) at (1,2) {$v_1$};
    \node[vertex] (v_2) at (1,0) {$v_2$};
    \node[vertex] (v_3) at (2.5,2)   {$v_3$};
    \node[vertex] (v_4) at (2.5,0)   {$v_4$};
    \node[vertex] (t)   at (3.5,1) {$t$};

    \draw[->] (s) -- node[weight,above left] {$e_1$} (v_1);
    \draw[->] (s) -- node[weight,below left] {$e_2$} (v_2);
    \draw[->] (v_1) -- node[weight,above] {$e_3$} (v_3);
    \draw[->] (v_1) -- node[weight,above right] {$e_4$} (v_4);
    \draw[->] (v_2) -- node[weight,below] {$e_5$} (v_4);
    \draw[->] (v_3) -- node[weight,above right] {$e_6$} (t);
    \draw[->] (v_4) -- node[weight,below right] {$e_7$} (t);
\end{tikzpicture}
    \caption{PERT network $\PN$.}\label{fig:reduction-pert}
\end{subfigure}\hfill%
\begin{subfigure}[t]{0.31\textwidth}
    \centering
    \begin{tikzpicture}[>=latex,scale=0.9,every node/.style={scale=0.8}]
    \node[place,label={left:$[u,e]$}] (u_e) at (0,0) {};
    \node[place,label={right:$[e,v]$}] (e_v) at (2,0) {};
    \node [transition,label={above:\textcolor{tcolor}{$t_{e,0}$}}] (t_e_0) at (1,0.3) {$\ttime{0}$}
      edge [post] (e_v);
    \draw[pre] (t_e_0) -- node[above] {$\tweight{\overline{p_e}}$} (u_e);
    \node [transition,label={below:\textcolor{tcolor}{$t_{e,1}$}}] (t_e_1) at (1,-0.3) {$\ttime{1}$}
      edge [post] (e_v);
    \draw[pre] (t_e_1) -- node[below] {$\tweight{p_e}$} (u_e);
\end{tikzpicture}
    \caption{Gadget for $e = (u,v)$ with rational weights $p_e,\overline{p_e}$.}\label{fig:reduction-weight-rational}
\end{subfigure}\hfill%
\begin{subfigure}[t]{0.38\textwidth}
    \centering
    \begin{tikzpicture}[>=latex,scale=0.8,every node/.style={scale=0.8}]
    \node[place,label={left:$[u,e]$}] (u_e) at (2,-2) {};

    \node[place,label={above:$q_1$}] (q_1) at (2,0) {};
    \node[place,label={above:$q_2$}] (q_2) at (4,0) {};
    \node[place,label={above:$q_3$}] (q_3) at (6,0) {};

    \node[place,label={left:$[e,v]$}] (e_v) at (4,-2) {};

    \node [transition,marked,label={left:\textcolor{tcolor}{$a_{e,0}$}}] (a_e_0) at (2,-1) {$\ttime{0}$}
      edge [pre] (u_e)
      edge [post] (q_1);

    \node [transition,label={above:\textcolor{tcolor}{$b_{e,1}$}}] (b_e_1) at (3,0) {$\ttime{0}$}
      edge [pre] (q_1)
      edge [post] (q_2);
    \node [transition,marked,label={below:\textcolor{tcolor}{$a_{e,1}$}}] (a_e_1) at (3,-1) {$\ttime{1}$}
      edge [pre] (q_1)
      edge [post] (e_v);

    \node [transition,label={above:\textcolor{tcolor}{$b_{e,2}$}}] (b_e_2) at (5,0) {$\ttime{0}$}
      edge [pre] (q_2)
      edge [post] (q_3);
    \node [transition,marked,label={[xshift=-4mm]above:\textcolor{tcolor}{$a_{e,2}$}}] (a_e_2) at (4,-1) {$\ttime{0}$}
      edge [pre] (q_2)
      edge [post] (e_v);

    \node [transition,label={below:\textcolor{tcolor}{$b_{e,3}$}}] (b_e_3) at (6,-1) {$\ttime{0}$}
      edge [pre] (q_3)
      edge [post] (e_v);
    \node [transition,marked,label={above:\textcolor{tcolor}{$a_{e,3}$}}] (a_e_3) at (5,-1) {$\ttime{1}$}
      edge [pre] (q_3)
      edge [post] (e_v);
\end{tikzpicture}
    \caption{Equivalent gadget for $e$ with weights 1 for $p_e = \nicefrac{5}{8} = {(\textcolor{markedtext}{0.101})}_2$.}\label{fig:reduction-weight-one}
\end{subfigure}
\begin{subfigure}[t]{\textwidth}
    \centering
    \begin{tikzpicture}[>=latex,scale=0.9,every node/.style={scale=0.8}]
    \node[place,label={above:$i$},tokens=1] (i) at (0.5,1) {};

    \node[place,label={left:$[s,e_1]$}] (s_e_1) at (2,2) {};
    \node[place,label={above:$[e_1,v_1]$}] (e_1_v_1) at (4,2) {};
    \node [transition,label={above:\textcolor{tcolor}{$t_{e_1,0}$}}] (t_e_1_0) at (3,2.3) {$\ttime{0}$}
      edge [post] (e_1_v_1);
    \draw[pre] (t_e_1_0) -- node[above] {$\tweight{\overline{p_{e_1}}}$} (s_e_1);
    \node [transition,label={below:\textcolor{tcolor}{$t_{e_1,1}$}}] (t_e_1_1) at (3,1.7) {$\ttime{1}$}
      edge [post] (e_1_v_1);
    \draw[pre] (t_e_1_1) -- node[below] {$\tweight{p_{e_1}}$} (s_e_1);

    \node[place,label={left:$[s,e_2]$}] (s_e_2) at (2,0) {};
    \node[place,label={above:$[e_2,v_2]$}] (e_2_v_2) at (4,0) {};
    \node [transition,label={above:\textcolor{tcolor}{$t_{e_2,0}$}}] (t_e_2_0) at (3,0.3) {$\ttime{0}$}
      edge [post] (e_2_v_2);
    \draw[pre] (t_e_2_0) -- node[above] {$\tweight{\overline{p_{e_2}}}$} (s_e_2);
    \node [transition,label={below:\textcolor{tcolor}{$t_{e_2,1}$}}] (t_e_2_1) at (3,-0.3) {$\ttime{1}$}
      edge [post] (e_2_v_2);
    \draw[pre] (t_e_2_1) -- node[below] {$\tweight{p_{e_2}}$} (s_e_2);

    \node[place,label={left:$[v_1,e_3]$}] (v_1_e_3) at (5.5,3) {};
    \node[place,label={above:$[e_3,v_3]$}] (e_3_v_3) at (7.5,3) {};
    \node [transition,label={above:\textcolor{tcolor}{$t_{e_3,0}$}}] (t_e_3_0) at (6.5,3.3) {$\ttime{0}$}
      edge [post] (e_3_v_3);
    \draw[pre] (t_e_3_0) -- node[above] {$\tweight{\overline{p_{e_3}}}$} (v_1_e_3);
    \node [transition,label={below:\textcolor{tcolor}{$t_{e_3,1}$}}] (t_e_3_1) at (6.5,2.7) {$\ttime{1}$}
      edge [post] (e_3_v_3);
    \draw[pre] (t_e_3_1) -- node[below] {$\tweight{p_{e_3}}$} (v_1_e_3);

    \node[place,label={left:$[v_1,e_4]$}] (v_1_e_4) at (5.5,1) {};
    \node[place,label={above:$[e_4,v_4]$}] (e_4_v_4) at (7.5,1) {};
    \node [transition,label={above:\textcolor{tcolor}{$t_{e_4,0}$}}] (t_e_4_0) at (6.5,1.3) {$\ttime{0}$}
      edge [post] (e_4_v_4);
    \draw[pre] (t_e_4_0) -- node[above] {$\tweight{\overline{p_{e_4}}}$} (v_1_e_4);
    \node [transition,label={below:\textcolor{tcolor}{$t_{e_4,1}$}}] (t_e_4_1) at (6.5,0.7) {$\ttime{1}$}
      edge [post] (e_4_v_4);
    \draw[pre] (t_e_4_1) -- node[below] {$\tweight{p_{e_4}}$} (v_1_e_4);

    \node[place,label={[yshift=-1mm]left:$[v_1,e_5]$}] (v_2_e_5) at (5.5,-1) {};
    \node[place,label={below:$[e_5,v_4]$}] (e_5_v_4) at (7.5,-1) {};
    \node [transition,label={above:\textcolor{tcolor}{$t_{e_5,0}$}}] (t_e_5_0) at (6.5,-0.7) {$\ttime{0}$}
      edge [post] (e_5_v_4);
    \draw[pre] (t_e_5_0) -- node[above] {$\tweight{\overline{p_{e_5}}}$} (v_2_e_5);
    \node [transition,label={below:\textcolor{tcolor}{$t_{e_5,1}$}}] (t_e_5_1) at (6.5,-1.3) {$\ttime{1}$}
      edge [post] (e_5_v_4);
    \draw[pre] (t_e_5_1) -- node[below] {$\tweight{p_{e_5}}$} (v_2_e_5);

    \node[place,label={[xshift=-1mm]below:$[v_3,e_6]$}] (v_3_e_6) at (9,2) {};
    \node[place,label={above:$[e_6,t]$}] (e_6_t) at (11,2) {};
    \node [transition,label={above:\textcolor{tcolor}{$t_{e_6,0}$}}] (t_e_6_0) at (10,2.3) {$\ttime{0}$}
      edge [post] (e_6_t);
    \draw[pre] (t_e_6_0) -- node[above] {$\tweight{\overline{p_{e_6}}}$} (v_3_e_6);
    \node [transition,label={below:\textcolor{tcolor}{$t_{e_6,1}$}}] (t_e_6_1) at (10,1.7) {$\ttime{1}$}
      edge [post] (e_6_t);
    \draw[pre] (t_e_6_1) -- node[below] {$\tweight{p_{e_6}}$} (v_3_e_6);

    \node[place,label={[xshift=-1mm]below:$[v_4,e_7]$}] (v_4_e_7) at (9,0) {};
    \node[place,label={below:$[e_7,t]$}] (e_7_t) at (11,0) {};
    \node [transition,label={above:\textcolor{tcolor}{$t_{e_7,0}$}}] (t_e_7_0) at (10,0.3) {$\ttime{0}$}
      edge [post] (e_7_t);
    \draw[pre] (t_e_7_0) -- node[above] {$\tweight{\overline{p_{e_7}}}$} (v_4_e_7);
    \node [transition,label={below:\textcolor{tcolor}{$t_{e_7,1}$}}] (t_e_7_1) at (10,-0.3) {$\ttime{1}$}
      edge [post] (e_7_t);
    \draw[pre] (t_e_7_1) -- node[below] {$\tweight{p_{e_6}}$} (v_4_e_7);

    \node[place,label={above:$o$}] (o) at (12.5,1) {};

    \node [transition,label={above:\textcolor{tcolor}{$t_s$}}] (t_s) at (1.25,1) {$\ttime{0}$}
      edge [pre]  (i)
      edge [post] (s_e_1)
      edge [post] (s_e_2);
    \node [transition,label={above:\textcolor{tcolor}{$t_{v_1}$}}] (t_v_1) at (4.75,2) {$\ttime{0}$}
      edge [pre]  (e_1_v_1)
      edge [post] (v_1_e_3)
      edge [post] (v_1_e_4);
    \node [transition,label={above:\textcolor{tcolor}{$t_{v_2}$}}] (t_v_2) at (4.75,-0.5) {$\ttime{0}$}
      edge [pre]  (e_2_v_2)
      edge [post] (v_2_e_5);
    \node [transition,label={above:\textcolor{tcolor}{$t_{v_3}$}}] (t_v_1) at (8.25,2.5) {$\ttime{0}$}
      edge [pre]  (e_3_v_3)
      edge [post] (v_3_e_6);
    \node [transition,label={above:\textcolor{tcolor}{$t_{v_4}$}}] (t_v_4) at (8.25,0) {$\ttime{0}$}
      edge [pre]  (e_4_v_4)
      edge [pre]  (e_5_v_4)
      edge [post] (v_4_e_7);
    \node [transition,label={above:\textcolor{tcolor}{$t_t$}}] (t_t) at (11.75,1) {$\ttime{0}$}
      edge [pre]  (e_6_t)
      edge [pre]  (e_7_t)
      edge [post] (o);

\end{tikzpicture}
    \caption{Timed probabilistic workflow net $\W_\PN$.}\label{fig:reduction-net}
\end{subfigure}
\caption{A PERT network and its corresponding timed probabilistic workflow net.
    The weight $\overline{p}$ is short for $1-p$. Transitions without annotations
    have weight $1$.
}\label{fig:reduction-pert-net}
\end{figure}

The result of applying this construction to the PERT network from
Figure~\ref{fig:reduction-pert} is shown in Figure~\ref{fig:reduction-net}. It
is easy to see that this workflow net is sound, as from any reachable marking,
we can fire enabled transitions corresponding to the edges and vertices of the
PERT network in the topological order of the graph, eventually firing $t_t$ and
reaching $\mO$. The net is also acyclic and free-choice.

\begin{restatable}{lemma}{lemreductionpartone}\label{lem:reduction-part-one}
    Let $\PN$ be a two-state stochastic PERT network
    and let $\W_\PN$ be its corresponding TPWN
    by the construction above.
    Then $ET_{\W_\PN} = \Ex(PD(\PN))$.
\end{restatable}

\subsubsection{Second reduction: 0/1 times, 0/1 weights.}
The network constructed this way already uses times $0$ and $1$, however the
weights still use arbitrary rational numbers. We now replace the gadget nets
from Figure~\ref{fig:reduction-weight-rational} by equivalent nets where all
transitions have weight $1$. The idea is to use the binary encoding of the
probabilities $p_e$, deciding if the time is $0$ or $1$ by a sequence of coin
flips. We assume that $p_e = \sum_{i=0}^k 2^{-i} p_i$ for some $k \in \N$ and
$p_i \in \set{0, 1}$ for $0 \le i \le k$. The replacement is shown in
Figure~\ref{fig:reduction-weight-one} for $p_e = \nicefrac{5}{8} = {(0.101)}_2$.

\subsubsection{Approximating the expected time is \sharpp-hard.}
We show that computing an $\epsilon$-approximation for $ET_\W$ is
\sharpp-hard~\cite{Hagstrom88,ProvanB83}.

\begin{restatable}{theorem}{thmapproximationsharpphard}\label{thm:approximation-sharpp-hard}
    The following problem is \sharpp{}-hard:
\begin{quote}
{\bf Given}:
    A sound, acyclic and free-choice TPWN $\W$ where all transitions $t$
    satisfy $w(t) = 1$, $\tau(t) \in \set{ 0, 1 }$ and
    $\abs{\postset{(\preset{t})}} \le 2$, and an $\epsilon > 0$.\\
{\bf Compute}:
    A rational $r$ such that $r - \epsilon < ET_W < r + \epsilon$.
\end{quote}
\end{restatable}

\section{Experimental evaluation}\label{sec:experiments}
We have implemented our algorithm to compute the expected time of a TPWN as a
package of the tool \prom{}\footnote{\url{http://www.promtools.org/}}. It is
available via the package manager of the latest nightly build under the package
name \texttt{WorkflowNetAnalyzer}.

We evaluated the algorithm on two different benchmarks. All experiments in
this section were run on the same machine equipped with an Intel Core i7-6700K
CPU and \SI{32}{\giga\byte} of RAM\@. We measure the actual runtime of the
algorithm, split into construction of the Markov chain and solving the linear
equation system, and exclude the time overhead due to starting \prom{} and
loading the plugin.

\subsection{IBM benchmark}

We evaluated the tool on a set of 1386 workflow nets extracted from a
collection of five libraries of industrial business processes modeled in the
IBM WebSphere Business Modeler~\cite{FFJKLVW09}. All of the 1386 nets in the
benchmark libraries are free-choice and therefore confusion-free. We selected
the sound and 1-safe nets among them, which are 642 nets. Out of these, 409
are marked graphs, i.e.~the size of any conflict set is 1. Out of the remaining
233 nets, 193 are acyclic and 40 cyclic.

As these nets do not come with probabilistic or time information, we annotated
transitions with integer weights and times chosen uniformely from different
intervals:
\begin{enumerate*}[label=(\arabic*)]
    \item $w(t) = \tau(t) = 1$,
    \item $w(t),\tau(t) \in [1,10^3]$ and
    \item $w(t),\tau(t) \in [1,10^6]$.
\end{enumerate*}
For each interval, we annotated the transitions of each net with random weights
and times, and computed the expected time of all 642 nets.

For all intervals, we computed the expected time for any net in less than
\SI{50}{\ms}. The analysis time did not differ much for different intervals.
The solving time for the linear equation system is on average $5\%$ of the
total analysis time, and at most $68\%$. The results for the nets with the
longest analysis times are given in Table~\ref{tab:ibm-select-results}. They
show that even for nets with a huge state space, thanks to the earliest-first
scheduler, only a small number of reachable markings is explored.

\begin{table}[t]
    \setlength{\tabcolsep}{3pt}
    \begin{center}
    \begin{tabular}[t]{lrrrrrrrrrrr}
        \toprule
        & \multicolumn{4}{c}{Net info \& size} & \multicolumn{3}{c}{Analysis time (\si{\ms})}  & \multicolumn{3}{c}{$\abs{\vec{X}}$}\\
        \cmidrule(r){2-5}
        \cmidrule(r){6-8}
        \cmidrule(r){9-11}
        Net & cyclic & $\abs{P}$ & $\abs{T}$ & $\abs{\reachwn[N]}$ & $[1]$ & $[10^3]$ & $[10^6]$ & [1] & $[10^3]$ & $[10^6]$ \\
        \midrule
        m1.s30\_s703 & no & 264 & 286 & 6117 & 40.3 & 44.6 & 43.8 & 304 & 347 & 347 \\
        m1.s30\_s596 & yes & 214 & 230 & 623 & 21.6 & 24.4 & 23.6 & 208 & 232 & 234 \\
        b3.s371\_s1986 & no & 235 & 101 & $2\cdot10^{17}$ & 16.8 & 16.4 & 16.5 & 101 & 102 & 102 \\
        b2.s275\_s2417 & no & 103 & 68 & 237626 & 14.2 & 17.8 & 15.9 & 355 & 460 & 431 \\
        \bottomrule
    \end{tabular}
    \end{center}
    \caption{Analysis times and size of the state space $\abs{\vec{X}}$ for the 4
        nets with the highest analysis times, given for each of the three
        intervals $[1],[10^3],[10^6]$ of possible times. Here, $\abs{\reachwn[N]}$
        denotes the number of reachable markings of the net.
    }\label{tab:ibm-select-results}
\end{table}

\subsection{Process Mining Case Study}

As a second benchmark, we evaluated the algorithm on a model of a loan
application process. We used the data from the BPI Challenge
2017~\cite{BPI2017Data}, an event log containing 31509 cases provided by a
financial institute, and took as a model of the process the final net from the
report of the winner of the academic category~\cite{BPI2017AcademicWinner}, a
simple model with high fitness and precision w.r.t.~the event log.

Using the \prom{} plugin ``Multi-perspective Process
Explorer''~\cite{DBLP:conf/bpm/MannhardtLR15} we annotated each transition with
waiting times and each transition in a conflict set with a local percentage of
traces choosing this transition when this conflict set is enabled. The net with
mean times and weights as percentages is displayed in
Fig.~\ref{fig:bpi-net}.

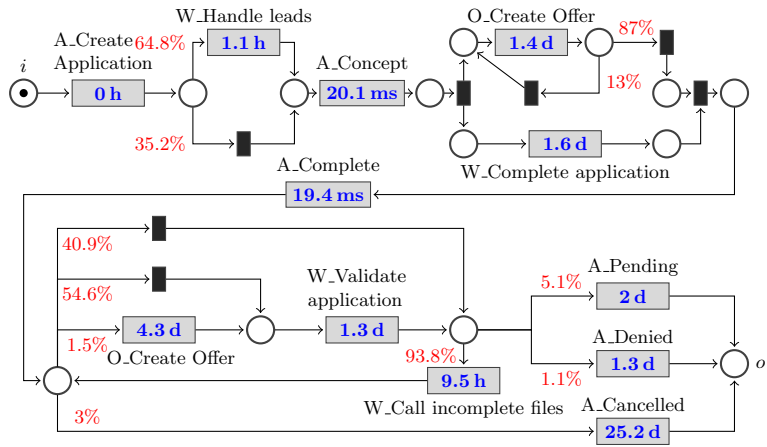
\begin{figure}
\centering
\begin{tikzpicture}[scale=0.9,every node/.style={scale=0.8}]
    \node[place,label={above:$i$},tokens=1] (i) at (-0.5,0.25) {};
    \node[place] (p1) at (2,0.25) {};
    \node[place] (p2) at (3.5,0.25) {};
    \node[place] (p3) at (5.5,0.25) {};
    \node[place] (p4) at (6,1.0) {};
    \node[place] (p5) at (6,-0.5) {};
    \node[place] (p6) at (8,1.0) {};
    \node[place] (p7) at (9,-0.5) {};
    \node[place] (p8) at (9,0.25) {};
    \node[place] (p9) at (10,0.25) {};

    \node[place] (p10) at (0,-4) {};
    \node[place] (p11) at (3,-3.25) {};
    \node[place] (p12) at (6,-3.25) {};
    \node[place,label={right:$o$}] (o) at (10,-3.75) {};

    \node [t-timed] (a-create-application) at (0.75,0.25) {\ttime{\SI{0}{\hour}}}
      edge [pre] (i)
      edge [post] (p1);
      \node [above=0mm of a-create-application,text width=1.8cm]{A\_Create Application};

      \node [t-timed,label={above:W\_Handle leads}] (w-handle-leads) at (2.75,1.0) {\ttime{\SI{1.1}{\hour}}};
      \draw[post] (w-handle-leads) -| (p2);
      \draw[pre] (w-handle-leads) -| node[left] {$\tweight{64.8\%}$} (p1);

      \node [t-immediate] (tau1) at (2.75,-0.5) {};
      \draw[post] (tau1) -| (p2);
      \draw[pre] (tau1) -| node[left] {$\tweight{35.2\%}$} (p1);

      \node [t-timed,label={above:A\_Concept}] (a-concept) at (4.5,0.25) {\ttime{\SI{20.1}{\ms}}}
      edge [pre] (p2)
      edge [post] (p3);

      \node [t-immediate] (tau2) at (6,0.25) {}
      edge [pre] (p3)
      edge [post] (p4)
      edge [post] (p5);

      \node [t-timed,label={above:O\_Create Offer}] (o-create-offer1) at (7,1.0) {\ttime{\SI{1.4}{\day}}}
      edge [pre] (p4)
      edge [post] (p6);

      \node [t-immediate] (tau3) at (7,0.25) {}
      edge [post] (p4);
      \draw[pre] (tau3) -| node[right,yshift=2mm] {$\tweight{13\%}$} (p6);

      \node [t-immediate] (tau4) at (9,1.0) {}
      edge [post] (p8);
      \draw[pre] (tau4) -- node[above] {$\tweight{87\%}$} (p6);

      \node [t-timed,label={below:W\_Complete application}] (w-complete-application) at (7.5,-0.5) {\ttime{\SI{1.6}{\day}}}
      edge [pre] (p5)
      edge [post] (p7);

      \node [t-immediate] (tau5) at (9.5,0.25) {}
      edge [pre] (p8)
      edge [post] (p9);
      \draw[pre] (tau5) |- (p7);

      \node [t-timed,label={above:A\_Complete}] (a-complete) at (4,-1.25) {\ttime{\SI{19.4}{\ms}}};
      \draw[pre] (a-complete) -| (p9);
      \draw[post] (a-complete) -- (-0.5,-1.25) |- (p10);

      \node [t-timed] (o-create-offer2) at (1.5,-3.25) {\ttime{\SI{4.3}{\day}}}
      edge [post] (p11);
      \draw[pre] (o-create-offer2) -| node[below,xshift=5mm] {$\tweight{1.5\%}$} (p10);
      \node [below=0mm of o-create-offer2,xshift=2mm] {O\_Create Offer};

      \node [t-immediate] (tau6) at (1.5,-2.5) {};
      \draw[post] (tau6) -| (p11);
      \draw[pre] (tau6) -| node[below,xshift=5mm] {$\tweight{54.6\%}$} (p10);

      \node [t-immediate] (tau7) at (1.5,-1.75) {};
      \draw[post] (tau7) -| (p12);
      \draw[pre] (tau7) -| node[below,xshift=5mm] {$\tweight{40.9\%}$} (p10);

      \node [t-timed] (w-validate-application) at (4.5,-3.25) {\ttime{\SI{1.3}{\day}}}
      edge [pre] (p11)
      edge [post] (p12);
      \node [above=0mm of w-validate-application,text width=1.8cm]{W\_Validate application};

    \node [t-timed,label={below:W\_Call incomplete files}] (w-call-incomplete-files) at (6,-4) {\ttime{\SI{9.5}{\hour}}}
      edge [post] (p10);
      \draw[pre] (w-call-incomplete-files) -- node[left] {$\tweight{93.8\%}$} (p12);

      \node [t-timed,label={above:A\_Pending}] (a-pending) at (8.5,-2.75) {\ttime{\SI{2}{\day}}};
      \draw[post] (a-pending) -| (o);
      \draw[pre] (a-pending) -| node[above,xshift=5mm] {$\tweight{5.1\%}$} (7,-3.25) -- (p12);

      \node [t-timed,label={above:A\_Denied}] (a-denied) at (8.5,-3.75) {\ttime{\SI{1.3}{\day}}}
      edge [post] (o);
      \draw[pre] (a-denied) -| node[below,xshift=5mm] {$\tweight{1.1\%}$} (7,-3.25) -- (p12);

      \node [t-timed,label={above:A\_Cancelled}] (a-cancelled) at (8.5,-4.75) {\ttime{\SI{25.2}{\day}}};
      \draw[pre] (a-cancelled) -| node[above,xshift=5mm] {$\tweight{3\%}$} (p10);
      \draw[post] (a-cancelled) -| (o);

\end{tikzpicture}
\caption{Net from~\cite{BPI2017AcademicWinner} of process for personal loan
applications in a financial institute, annotated with mean waiting times and
local trace weights. Black transitions are invisible transitions not appearing
in the event log with time 0.
}\label{fig:bpi-net}
\end{figure}

\begin{table}
    \setlength{\tabcolsep}{3pt}
    \begin{center}
        \begin{tabular}[t]{lrr@{\hspace{0.8mm}}rrrrrrrrr}
        \toprule
        & & & & & \multicolumn{3}{c}{Analysis time} \\
        \cmidrule(r){6-8}
        Distribution        & $\abs{T}$ & \multicolumn{2}{c}{$ET_W$} & $\abs{\vec{X}}$ & Total & Construction & Solving \\
        \midrule
        Deterministic            &   19 & \SI{24}{\day}&\SI{1}{\hour}  &     33 &  \SI{40}{\ms} &  \SI{18}{\ms}      & \SI{22}{\ms} \\
        \midrule
        Histogram/\SI{12}{\hour} &  141 & \SI{24}{\day}&\SI{18}{\hour} &   4054 & \SI{244}{\ms} & \SI{232}{\ms}      & \SI{12}{\ms} \\
        Histogram/\SI{6}{\hour}  &  261 & \SI{24}{\day}&\SI{21}{\hour} &  15522 & \SI{2.1}{\s}  & \SI{1.8}{\s}       & \SI{0.3}{\s}  \\
        Histogram/\SI{4}{\hour}  &  375 & \SI{24}{\day}&\SI{22}{\hour} &  34063 &  \SI{10}{\s}  & \SI{6}{\s}         & \SI{4}{\s}  \\

        Histogram/\SI{2}{\hour}  &  666 & \SI{24}{\day}&\SI{23}{\hour} & 122785 & \SI{346}{\s}  & \SI{52}{\s}        & \SI{294}{\s}  \\
        Histogram/\SI{1}{\hour}  & 1117 &              &           --- & 422614 &        ---    & \SI{12.7}{\minute} & \memout{}  \\
        \bottomrule
    \end{tabular}
    \end{center}
    \caption{Expected time, analysis time and state space size for the net in
        Fig.~\ref{fig:bpi-net} for various distributions, where \memout{}
        denotes reaching the memory limit.
    }\label{tab:bpi-results}
\end{table}

For a first analysis, we simply set the execution time of each transition
deterministically to its mean waiting time. However, note that the two
transitions ``O\_Create Offer'' and ``W\_Complete application'' are executed in
parallel, and therefore the distribution of their execution times influences
the total expected time. Therefore we also annotated these two transitions
with a histogram of possible execution times from each case. Then we split them
up into multiple transitions by grouping the times into buckets of a given
interval size, where each bucket creates a transition with an execution time
equal to the beginning of the interval, and a weight equal to the number of
cases with a waiting time contained in the interval. The times for these
transitions range from 6 milliseconds to 31 days. As bucket sizes we chose
$12,6,4,2$ and $1$ hour(s). The net always has $14$ places and $15$ reachable
markings, but a varying number of transitions depending on the chosen bucket
size. For the net with the mean as the deterministic time and for the nets
with histograms for each bucket size, we then analyzed the expected execution
time using our algorithm.

The results are given in Table~\ref{tab:bpi-results}. They show that using the
complete distribution of times instead of only the mean can lead to much more
precise results. When the linear equation system becomes very large, the
solver time dominates the construction time of the system. This may be because
we chose to use an exact solver for sparse linear equation systems. In the
future, this could possibly be improved by using an approximative iterative
solver.

\section{Conclusion}\label{sec:conclusion}
We have shown that computing the expected time to termination of a
probabilistic workflow net in which transition firings have deterministic
durations is \sharpp{}-hard. This is the case even if the net is free-choice,
and both probabilities and times can be written down with a constant number of
bits. So, surprisingly, computing the expected time is much harder than
computing the expected cost, for which there is a polynomial
algorithm~\cite{PEVA}.

We have also presented an exponential algorithm for computing the expected time
based on earliest-first schedulers. Its performance depends crucially on the
maximal size of conflict sets that can be concurrently enabled. In the most
popular suite of industrial benchmarks this number turns out to be small. So,
very satisfactorily, the expected time of any of these benchmarks, some of
which have hundreds of transitions, can still be computed in milliseconds.

\smallskip
\noindent \textbf{Acknowledgements.} We thank Hagen V\"olzer for input on the
implementation and choice of benchmarks.

\bibliographystyle{splncs04}
\bibliography{references}

\newpage
\appendix
\section{Appendix}\label{sec:appendix}
\subsection{Additional preliminaries}
\label{app-prelim}

\medskip

\noindent \textbf{Workflow Nets.} A \emph{marking} of a workflow net is a
function $M \colon P \to \N$, representing the number of tokens in each place.
A transition $t$ is \emph{enabled} at a marking $M$ if for all $p \in
\preset{t}$, we have $M(p) \ge 1$. If $t$ is enabled at $M$, it may
\emph{occur}, leading to a marking $M'$ obtained by removing one token from
each place of $\preset{t}$ and then adding one token to each place of
$\postset{t}$. We denote this by $M \trans{t} M'$. Formally, $M'$ is defined
by $M'(p) = M(p) - 1$ if $p \in \preset{t} \setminus \postset{t}$, $M'(p) =
M(p) + 1$ if $p \in \postset{t} \setminus \preset{t}$, and $M'(p) = M(p)$
otherwise. Let $\sigma = t_1 t_2 \ldots t_n$ be a sequence of transitions.
For a marking $M_0$, $\sigma$ is a \emph{firing} or \emph{occurrence sequence}
if $M_0 \trans{t_1} M_1 \trans{t_2} \ldots \trans{t_n} M_n$ for some markings
$M_1,\ldots,M_n$. We say that $M_n$ is reachable from $M_0$ by $\sigma$ and
denote this by $M_0 \trans{\sigma} M_n$.

Fig.~\ref{fig:example1} shows four workflow nets. We use them to illustrate
the definitions of Section~\ref{sec/preliminaries}.

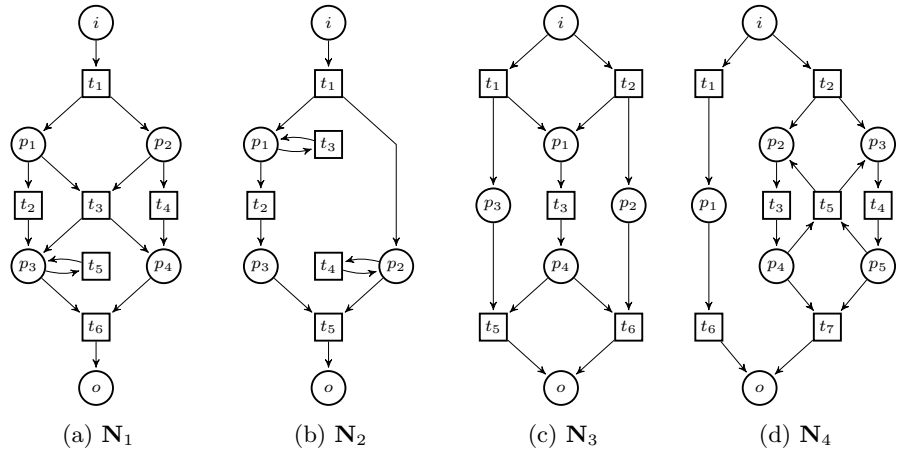
\begin{figure}
\centering
\begin{subfigure}[t]{0.24\textwidth}
    \centering
    \scalebox{0.9}{\begin{tikzpicture}[>=stealth',bend angle=45,auto]
	\tikzstyle{every node}=[font=\scriptsize]
	\tikzstyle{place}=[circle,thick,draw=black,fill=white,minimum size=5mm,inner sep=0mm]
	\tikzstyle{transition}=[rectangle,thick,draw=black,fill=white,minimum size=4mm,inner sep=0mm]
	\tikzstyle{every label}=[black]

	\node [place] at (0,0) (c0){$i$};
	\node [place] at (-1,-1.8) (c1){$p_1$};
	\node [place] at (1,-1.8) (c2){$p_2$};
    \node [place] at (-1,-3.6) (c3){$p_3$};
	\node [place] at (1,-3.6) (c4){$p_4$};
	\node [place] at (0,-5.4) (o){$o$};

	\node [transition] at (0,-0.9) (t1) {$t_1$}
	edge [pre] (c0)
	edge [post] (c1)
	edge [post] (c2);

	\node [transition] at (-1,-2.7) (t2) {$t_2$}
	edge [pre] (c1)
	edge [post] (c3);

	\node [transition] at (0,-2.7) (t3) {$t_3$}
	edge [pre] (c1)
    edge [pre] (c2)
    edge [post] (c3)
	edge [post] (c4);

    \node [transition] at (1,-2.7) (t4) {$t_4$}
	edge [pre] (c2)
	edge [post] (c4);

    \node [transition] at (0,-3.6) (t5) {$t_5$}
	edge [pre,bend left=15] (c3)
	edge [post,bend right=15] (c3);

    \node [transition] at (0,-4.5) (t6) {$t_6$}
	edge [pre] (c3)
    edge [pre] (c4)
	edge [post] (o);
\end{tikzpicture}}
    \caption{$\Net_1$}\label{fig:example1-1}%
\end{subfigure}\hfill%
\begin{subfigure}[t]{0.24\textwidth}
    \centering
    \scalebox{0.9}{\begin{tikzpicture}[>=stealth',bend angle=45,auto]
	\tikzstyle{every node}=[font=\scriptsize]
	\tikzstyle{place}=[circle,thick,draw=black,fill=white,minimum size=5mm,inner sep=0mm]
	\tikzstyle{transition}=[rectangle,thick,draw=black,fill=white,minimum size=4mm,inner sep=0mm]
	\tikzstyle{every label}=[black]

	\node [place] at (0,0) (c0){$i$};
	\node [place] at (-1,-1.8) (c1){$p_1$};
    \node [place] at (1,-3.6) (c2){$p_2$};
	\node [place] at (-1,-3.6) (c4){$p_3$};
	\node [place] at (0,-5.4) (o){$o$};

	\node [transition] at (0,-0.9) (t1) {$t_1$}
	edge [pre] (c0)
	edge [post] (c1);
    \draw[post] (t1) -- (1,-1.8) -- (c2);

	\node [transition] at (-1,-2.7) (t2) {$t_2$}
	edge [pre] (c1)
	edge [post] (c3);

	\node [transition] at (0,-1.8) (t3) {$t_3$}
	edge [pre, bend left=15] (c1)
	edge [post, bend right=15] (c1);

	\node [transition] at (0.0,-3.6) (t4) {$t_4$}
	edge [pre, bend left=15] (c2)
	edge [post, bend right=15] (c2);

	\node [transition] at (0,-4.5) (t5) {$t_5$}
    edge [pre] (c2)
	edge [pre] (c3)
	edge [post] (o);
\end{tikzpicture}}
    \caption{$\Net_2$}\label{fig:example1-2}%
\end{subfigure}\hfill%
\begin{subfigure}[t]{0.24\textwidth}
    \centering
    \scalebox{0.9}{\begin{tikzpicture}[>=stealth',bend angle=45,auto]
	\tikzstyle{every node}=[font=\scriptsize]
	\tikzstyle{place}=[circle,thick,draw=black,fill=white,minimum size=5mm,inner sep=0mm]
	\tikzstyle{transition}=[rectangle,thick,draw=black,fill=white,minimum size=4mm,inner sep=0mm]
	\tikzstyle{every label}=[black]

	\node [place] at (0,0) (c0){$i$};
	\node [place] at (0,-1.8) (c1){$p_1$};
	\node [place] at (1,-2.7) (c2){$p_2$};
    \node [place] at (-1,-2.7) (c3){$p_3$};
	\node [place] at (0,-3.6) (c4){$p_4$};
	\node [place] at (0,-5.4) (o){$o$};

	\node [transition] at (1,-0.9) (t1) {$t_2$}
	edge [pre] (c0)
	edge [post] (c1)
	edge [post] (c2);

	\node [transition] at (-1,-0.9) (t2) {$t_1$}
	edge [pre] (c0)
    edge [post] (c1)
	edge [post] (c3);

	\node [transition] at (0,-2.7) (t4) {$t_3$}
	edge [pre] (c1)
	edge [post] (c4);

	\node [transition] at (1,-4.5) (t5) {$t_6$}
	edge [pre] (c2)
    edge [pre] (c4)
	edge [post] (o);

	\node [transition] at (-1,-4.5) (t6) {$t_5$}
	edge [pre] (c3)
	edge [pre] (c4)
	edge [post] (o);
\end{tikzpicture}}
    \caption{$\Net_3$}\label{fig:example1-3}%
\end{subfigure}\hfill%
\begin{subfigure}[t]{0.24\textwidth}
    \centering
    \scalebox{0.9}{\begin{tikzpicture}[>=stealth',bend angle=45,auto]
	\tikzstyle{every node}=[font=\scriptsize]
	\tikzstyle{place}=[circle,thick,draw=black,fill=white,minimum size=5mm,inner sep=0mm]
	\tikzstyle{transition}=[rectangle,thick,draw=black,fill=white,minimum size=4mm,inner sep=0mm]
	\tikzstyle{every label}=[black]

	\node [place] at (-.75,0) (c0){$i$};
	\node [place] at (-1.5,-2.7) (c1){$p_1$};
	\node [place] at (-0.5,-1.8) (c2){$p_2$};
	\node [place] at (1,-1.8) (c3){$p_3$};
	\node [place] at (-0.5,-3.6) (c4){$p_4$};
	\node [place] at (1,-3.6) (c5){$p_5$};
	\node [place] at (-.75,-5.4) (o){$o$};

	\node [transition] at (-1.5,-0.9) (t1) {$t_1$}
	edge [pre] (c0)
	edge [post] (c1);

	\node [transition] at (0.25,-0.9) (t2) {$t_2$}
	edge [pre] (c0)
	edge [post] (c2)
	edge [post] (c3);

	\node [transition] at (-0.5,-2.7) (t3) {$t_3$}
	edge [pre] (c2)
	edge [post] (c4);

	\node [transition] at (1,-2.7) (t4) {$t_4$}
	edge [pre] (c3)
	edge [post] (c5);

	\node [transition] at (0.25,-2.7) (t5) {$t_5$}
	edge [pre] (c4)
	edge [pre] (c5)
	edge [post] (c2)
	edge [post] (c3);

	\node [transition] at (-1.5,-4.5) (t6) {$t_6$}
	edge [pre] (c1)
	edge [post] (o);

	\node [transition] at (0.25,-4.5) (t7) {$t_7$}
	edge [pre] (c4)
	edge [pre] (c5)
	edge [post] (o);;
\end{tikzpicture}}
    \caption{$\Net_4$}\label{fig:example1-4}%
\end{subfigure}\hfill%
\caption{Four workflow nets.}\label{fig:example1}%
\end{figure}

\medskip

\noindent \textbf{Soundness and 1-safeness.} All workflow nets of
Fig.~\ref{fig:example1} are sound and 1-safe. Adding an arc to $\Net_4$ from
$p_4$ to $t_6$ makes it unsound, because after firing $t_1$ the final marking
cannot be reached. Adding an arc from $t_7$ to $p_1$ preserves soundness but
makes the workflow not $1$-safe.

\medskip

\noindent \textbf{Independence, concurrency, conflict. } The marking $M=\{p_1,
p_2\}$ of the workflow $\Net_1$ of Fig.~\ref{fig:example1-1} enables $t_2,
t_3, t_4$. At this marking, $t_2$ and $t_4$ are concurrent, $t_2$ and $t_3$ are
in conflict, and $t_3$ and $t_4$ are in conflict.

The runs $t_1 \, t_2 \, t_4 \, t_6$ and $t_1 \, t_4 \, t_2 \, t_6$ of $\Net_1$
are Mazurkiewicz equivalent, but they are not Mazurkiewicz equivalent to $t_1
\, t_3 \, t_6$. Transitions $t_3$ and $t_4$ of the workflow $\Net_2$ are
independent, and so all runs of $t_1 (t_3 + t_4)^* t_2 t_5$ containing the same
number of occurrences of $t_3$ and $t_4$ are Mazurkiewicz equivalent.

\medskip

\noindent \textbf{Confusion-freeness, free-choice workflows.} The conflict sets
at the marking $M=\{p_1, p_2\}$ of the workflow $\Net_1$ of
Fig.~\ref{fig:example1-1} are $C(t_2, M)=\{t_2, t_3\}$, $C(t_4, M)=\{t_3,
t_4\}$, and $C(t_3, M)=\{t_2, t_3, t_4\}$. The workflow nets $\Net_1$ and
$\Net_2$ of Figure~\ref{fig:example1} are not confusion-free. For example, in
$\Net_1$, transitions $t_2$ and $t_4$ are concurrent at $M=\{p_1, p_2\}$. We
have $C(t_2,M)= \{t_2, t_3\}$, but $C(t_2, M \setminus \preset{t_4})=C(t_2,
\{p_1\}) = \{t_2\}$.

The workflows $\Net_3$ and $\Net_4$ are confusion-free. $\Net_3$ is not
free-choice because $t_5 \in \postset{p_3} \cap \postset{p_4}$, but $t_6 \in
\postset{p_4} \setminus \postset{p_3}$. $\Net_4$ is free-choice.

\subsection{Markov Decision Processes}
\label{app-MDP}

A \emph{Markov Decision Process} (MDP) is a tuple $\mathcal{M} = (Q,q_0,{\it
Steps})$ where $Q$ is a finite set of states, $q_0\in Q$ is the initial state,
and ${\it Steps} \colon Q \ra 2^{dist(Q)}$ is the probability transition
function.

A transition from a state $q$ corresponds to first nondeterministically
choosing a probability distribution $\lambda \in {\it Steps}(q)$ and then
choosing the successor state $q'$ probabilistically according to $\lambda$.

A path is a finite or infinite non-empty sequence $\pi = q_0 \trans{\lambda_0}
q_1 \trans{\lambda_1} q_2 \ldots$ where $\lambda_i \in {\it Steps}(q_i)$ for
every $i \geq 0$. We denote by $\pi(i)$ the $i$-th state along $\pi$ (i.e.~the
state $q_i$), and by $\pi^i$ the prefix of $\pi$ ending at $\pi(i)$ (if it
exists). The last state of a finite path $\pi$ is denoted ${\it last }(\pi)$. A
{\em scheduler} is a function that maps every finite path $\pi$ to one of the
distributions of ${\it Steps}({\it last}(\pi))$.

Given a scheduler $S$, we let $\mathit{Paths}^S$ denote all infinite paths $\pi
= q_0 \trans{\lambda_0} q_1 \trans{\lambda_1} q_2 \ldots$ starting in $q_0$ and
satisfying $\lambda_{i} = S(\pi^i)$ for every $i \geq 0$. We define a
probability measure $\mathit{Prob}^S$ on $\mathit{Paths}^S$ in the usual way
using cylinder sets~\cite{KemenyDenumerable}.

\subsection{Missing proofs of Section~\ref{sec/time}: TPWNs have the same
expected time under all schedulers}
\label{sec:rewardconfusfree}

We prove Theorem~\ref{thm:expected-time}, showing that the expected time of a
TPWN is \emph{independent of the scheduler}. Intuitively, confusion-freeness
guarantees that the schedulers of a PWN only decide the order in which
concurrent transitions occur, but only from a formal point of view, without
connection with physical reality. Indeed, consider the workflow $\Net_2$ of
Figure \ref{fig:example1}. After firing $t_1$, a scheduler can choose to fire
$t_4$ arbitrarily often before choosing the conflict set $\{t_2,t_3\}$, and
eventually $\{t_4,t_5\}$. If $t_4$ has a positive time, then the expected time
is different for each scheduler.

The first part of the section gives a useful characterization of the expected
reward under a scheduler. The second part proves the theorem.

\subsubsection{A characterization of the expected reward under a scheduler}

The {\em probabilistic language} of a scheduler assigns to each finite firing
sequence $\sigma$ the probability of the cylinder of all paths of $\MDP_\W$
with $\Pi(\sigma)$, where $\Pi(\sigma)$ is the path of $\MDP_\W$ corresponding
to $\sigma$. Formally, the {\em probabilistic language $\nu_S$} of a scheduler
$S$ is the function $\nu_S \colon T^* \rightarrow \R_{\geq 0}$ defined by
$\nu_S(\sigma) = \mathit{Prob}^S(cyl^S(\Pi(\sigma)))$, where ${\it cyl}^S(\pi)$
denotes the cylinder of the infinite paths of ${\it Paths}^S$ that extend
$\pi$. A transition sequence $\sigma$ is \emph{compatible with $S$} if
$\nu_S(\sigma)>0$.

\begin{example}
Consider the scheduler $S$ for the net in Fig.~\ref{fig:example1-4} that always
chooses $\{t_3\}$ every time the marking $\{p_2, p_3\}$ is reached. Assume
$w(t_1) = 2$ and $w(t_3) = 3$. Then, for example, $\nu_S(t_2\, t_3) =
\nu_S(t_2\, t_3 \, t_4)=\nicefrac{3}{5}$, $\nu_S(t_2 \, t_4)=0$, and $\nu_S(t_1
\, t_6)= \nicefrac{2}{5}$.
\end{example}

We have the following characterization of the expected time with respect to a
scheduler $S$:

\begin{lemma}
\label{lem:charexpected}
Let $\W = (\Net ,w,\td)$ be a TPWN and let $S$ be a scheduler of $\W$.
\begin{align*}
ET^S_\W &=
\begin{cases}
\infty & \text{ if $\displaystyle \sum_{\sigma \in {\it Run}_\W} \nu_S(\sigma) < 1$ } \\
\displaystyle \sum_{\sigma \in {\it Run}_\W} \time(\sigma) \cdot \nu_S(\sigma) & \text{ otherwise}
\end{cases}
\end{align*}
\end{lemma}
\begin{proof}
By definition, the time taken by an infinite path of $\MDP_\W$ until it
reaches a state is finite if{}f the path eventually reaches the state.
Therefore, the probability of the paths of ${\it Paths}^S$ that take finite
time is equal to $\sum_{\sigma \in {\it Run}_\W} \nu_S(\sigma)$. If
$\sum_{\sigma \in {\it Run}_\W} \nu_S(\sigma) < 1$ then the paths taking
infinite time have nonzero probability, and so $ET^S_\W = \infty$. Otherwise
they have zero probability, and $ET^S_\W = \sum_{\sigma \in {\it Run}_\W}
\time(\sigma) \cdot \nu_S(\sigma)$.
\end{proof}

\subsubsection{Independence of the scheduler}

We start the proof of Theorem~\ref{thm:expected-time} with a lemma:

\begin{lemma}
\label{lem:equivValue}
Let $S_1, S_2$ be schedulers of a 1-safe and confusion-free PWN $\W = (\Net
,w,\td)$, and let $\sigma_1, \sigma_2$ be Mazurkiewicz-equivalent runs of $\W$
compatible with $S_1$ and $S_2$, respectively. Then
$\time(\sigma_1)=\time(\sigma_2)$ and $\nu_{S_1}(\sigma_1) =
\nu_{S_2}(\sigma_2)$.
\end{lemma}
\begin{proof}
Since Mazurkiewicz equivalence is the reflexive and transitive closure of
$\equiv_1$, it suffices to prove the result for the special case
$\sigma_1 \equiv_1 \sigma_2$.

Let $\sigma_1  = \tau \,  t_1\,  t_2 \, \tau'$ and $\sigma_1  = \tau \,  t_2\,
t_1 \, \tau'$ for independent transitions $t_1$ and $t_2$, and let
\begin{align*}
\mI &\trans{\tau} M \trans{t_1} M_1 \trans {t_2} M' \trans{\tau'} \mO
& &\text{and} &
\mI &\trans{\tau} M \trans{t_2} M_2 \trans {t_1} M' \trans{\tau'} \mO.
\end{align*}

As $t_1$ and $t_2$ are independent, we have $\preset{t_1} \cap \preset{t_2} =
\emptyset$, and both $t_1$ and $t_2$ are enabled at $M$. As the net is 1-safe,
we must have $\postset{t_1} \cap \preset{t_2} = \emptyset$, $\postset{t_2} \cap
\preset{t_1} = \emptyset$ and $\postset{t_1} \cap \postset{t_2} = \emptyset$,
as otherwise $M_1$, $M_2$ or $M'$, respectively, would contain more than one
token in a place in the intersection. Together, we have $(\preset{t_1} \cup
\postset{t_1}) \cap (\preset{t_2} \cup \postset{t_2}) = \emptyset$.

We have $\time(\sigma_1)=\time(\sigma_2)$ if and only if $\time(\tau t_1 t_2) =
\time(\tau t_2 t_1)$, given by:
\begin{align*}
    \time(\tau t_1 t_2)_p &=
    \begin{cases}
        \max_{q \in \preset{t_2}} \lambda(\tau t_1)_q + \tau(t_2) & \text{if }p \in \preset{t_2} \cup \postset{t_2} \\
        \lambda(\tau t_1)_p & \text{if }p \not\in \preset{t_2} \cup \postset{t_2}
    \end{cases} \\
    \time(\tau t_2 t_1)_p &=
    \begin{cases}
        \max_{q \in \preset{t_1}} \lambda(\tau t_2)_q + \tau(t_1) & \text{if }p \in \preset{t_1} \cup \postset{t_1} \\
        \lambda(\tau t_2)_p & \text{if }p \not\in \preset{t_1} \cup \postset{t_1}
    \end{cases}
\end{align*}
For any $p \in \preset{t_1} \cup \postset{t_1}$, we have $\lambda(\tau t_2)_p =
\lambda(\tau)_p$ and $\lambda(\tau t_1)_p = \max_{q \in \preset{t_1}}
\lambda(\tau)_q$, and for any $p \in \preset{t_2} \cup \postset{t_2}$, we have
$\lambda(\tau t_1)_p = \lambda(\tau)_p$ and $\lambda(\tau t_2)_p = \max_{q \in
\preset{t_2}} \lambda(\tau)_q$.  Thus we have
\begin{align*}
    \time(\tau t_2 t_1)_p &= \time(\tau t_1 t_2)_p =
    \begin{cases}
        \max_{q \in \preset{t_1}} \lambda(\tau)_q + \tau(t_1) & \text{if }p \in \preset{t_1} \cup \postset{t_1} \\
        \max_{q \in \preset{t_2}} \lambda(\tau)_q + \tau(t_2) & \text{if }p \in \preset{t_2} \cup \postset{t_2} \\
        \lambda(\tau)_p & \text{if }p \not\in (\preset{t_1} \cup \postset{t_1}) \cup (\preset{t_2} \cup \postset{t_2})
    \end{cases}\ .
\end{align*}

We have that $\nu_{S_1}(\sigma_1) = \nu_{S_1}(\sigma_2)$ if and only if
$\nu_{S_1}(\tau t_1 t_2) = \nu_{S_2}(\tau t_1 t_2)$, as the sequences then both
proceed with $\tau'$. $\nu_{S_1}(\tau t_1 t_2)$ is the probability of executing
first $\tau$ and then $t_1 t_2$ from $M$ under scheduler $S_1$, and
$\nu_{S_2}(\tau t_2 t_1)$ is the probability of executing first $\tau$ and then
$t_2 \, t_1$ from $M$ under scheduler $S_2$. As $\W$ is confusion free, the
conflict set of each transition in $\tau$ during its execution is uniquely
determined, so we have $\nu_{S_1}(\tau) = \nu_{S_2}(\tau)$. Further we have
\begin{align*}
\nu_{S_1}(\tau t_1 t_2) &= \nu_{S_1}(\tau) \cdot \frac{w(t_1)}{w(C(t_1, M))} \cdot \frac{w(t_2)}{w(C(t_2, M_1))} \\
\nu_{S_2}(\tau t_2 t_1) &= \nu_{S_2}(\tau) \cdot \frac{w(t_2)}{w(C(t_2, M))} \cdot \frac{w(t_1)}{w(C(t_1, M_2))}\ .
\end{align*}
Observe that $M_1 = (M \setminus \preset{t_1}) \cup \postset{t_1}$ and $M_2 =
(M \setminus \preset{t_2}) \cup \postset{t_2}$. Since $\W$ is confusion-free,
we have $C(t_1, M) = C(t_1, M_2)$, and $C(t_2, M) = C(t_2, M_1)$, and so
$\nu_{S_1}(\sigma_1) = \nu_{S_2}(\sigma_2)$.
\end{proof}

In~\cite{PEVA} we proved the following result:

\begin{lemma}
\label{lem:mazur2}
Let $\W$ be a confusion-free PWN and let $S$ be a scheduler of $\W$. For every
run $\sigma$ of $\W$ there is exactly one run $\tau \equiv \sigma$ compatible
with $S$.
\end{lemma}

We use it to prove that all schedulers have the same expected time. Further,
we show that the expected time of a TPWN is infinite if{}f the TPWN is sound.

\thmexpectedtime*
\begin{proof}
\noindent (1) Let $S_1$, $S_2$ be any two schedulers of $\W$. We prove
$ET^{S_1}_{\W} = ET^{S_2}_{\W}$. Let ${\cal R}_1 $ and ${\cal R}_2$ be the
sets of runs of $\W$ compatible with $S_1$ and $S_2$, respectively. By
Lemma~\ref{lem:mazur2}, there exists a map $\phi_{12} \colon {\cal R}_1
\rightarrow {\cal R}_2$ that assigns to every $\sigma \in {\cal R}_1$ the
unique run of ${\cal R}_2$ that is Mazurkiewicz equivalent to $\sigma$. We
proceed in three steps.

\subsubsection{Claim I} $\phi_{12}$ is bijective. \\
To prove injectivity, assume $\phi_{12}(\sigma_1) = \sigma_2 =
\phi_{12}(\sigma_1')$ for two different runs $\sigma_1, \sigma_1' \in {\cal
R}_1$. Then we have $\sigma_1 \equiv \sigma_2 \equiv \sigma'_1$, and so
$\sigma_2$ is Mazurkiewicz equivalent to two different runs of ${\cal R}_1$,
contradicting Lemma~\ref{lem:mazur2}. To prove surjectivity, let $\sigma_2 \in
{\cal R}_2$. By Lemma~\ref{lem:mazur2} there is a unique run $\sigma_1 \in
{\cal R}_1$ such that $\sigma_1 \equiv \sigma_2$, and so necessarily
$\phi_{12}(\sigma_1) = \sigma_2$.

\subsubsection{Claim II} $\displaystyle \sum_{\sigma \in {\it Run}_\W}
\time(\sigma) \cdot \nu_{S_1}(\sigma) = \sum_{\sigma \in {\it Run}_\W}
\time(\sigma) \cdot \nu_{S_2}(\sigma)$.

$$\begin{array}{rcll}
&   & \displaystyle \sum_{\sigma \in {\it Run}_\W} \time(\sigma) \cdot \nu_{S_1}(\sigma)  \\[0.5cm]
& = & \displaystyle \sum_{\sigma_1 \in {\cal R}_1} \time(\sigma_1) \cdot \nu_{S_1}(\sigma_1)
& \mbox{($\nu_{S_1}(\sigma)=0$ for every $\sigma \in {\it Run}_\W \setminus {\cal R}_1$)}\\[0.5cm]
& = & \displaystyle \sum_{\sigma_1 \in {\cal R}_1} r(\phi_{12}(\sigma_1))\cdot \nu_{S_1}(\phi_{12}(\sigma_1))
& \mbox{($\sigma_1 \equiv \phi_{12}(\sigma_1)$ and Lemma~\ref{lem:equivValue})}\\[0.5cm]
& = & \displaystyle \sum_{\sigma_2\in {\cal R}_2 }\time(\sigma)\cdot \nu_{S}(\sigma)
& \mbox{($\phi_{12}$ is bijective by Claim I)} \\[0.5cm]
& = & \displaystyle \sum_{\sigma \in {\it Run}_\W} \time(\sigma_2) \cdot \nu_{S_2}(\sigma)
& \mbox{($\nu_{S_2}(\sigma)=0$ for every $\sigma \in {\it Run}_\W \setminus {\cal R}_2$)}
\end{array}$$

\subsubsection{Claim III} $\displaystyle \sum_{\sigma \in {\it Run}_\W}
\nu_{S_1}(\sigma) = \sum_{\sigma \in {\it Run}_\W} \nu_{S_2}(\sigma)$.

\smallskip

\noindent This is proved by the same sequence of steps followed in Claim II.

\bigskip

\noindent We are now ready to show $ET^{S_1}_\W = ET^{S_2}_\W$. Consider two
cases.
\begin{itemize}
\item $\displaystyle \sum_{\sigma \in {\it Run}_\W} \nu_{S_1}(\sigma) < 1$.
    Then $ET^{S_1}_\W = \infty = ET^{S_2}_\W$ by Lemma~\ref{lem:charexpected}
    and Claim III.
\item $\displaystyle \sum_{\sigma \in {\it Run}_\W} \nu_{S_1}(\sigma) = 1$. By
    Lemma~\ref{lem:charexpected} and Claim II, we have $$ET^{S_1}_\W=
    \displaystyle \sum_{\sigma \in {\it Run}_\W} \time(\sigma_1) \cdot
    \nu_{S_1}(\sigma)  = \sum_{\sigma \in {\it Run}_\W}  \time(\sigma) \cdot
    \nu_{S_2}(\sigma) = ET^{S_2}_\W \ .$$
\end{itemize}

\medskip

\noindent (2) If $\W$ is unsound, then there is a firing sequence $\mI
\trans{\sigma} M$ such that $\mO$ is unreachable from $M$. So the reward of
every path extending $\Pi(\sigma)$ is infinite. Let $S$ be any scheduler such
that $\nu_S(\sigma) > 0$. Then the probability of the cylinder
$cyl^S(\Pi(\sigma))$ is nonzero. So $ET^S_\W = \infty$, and by (1) we have
$ET_\W = \infty$.

Assume now that $\W$ is sound. Let $S$ be a memoryless scheduler of $\W$, and
let $\MDP^S_\W$ be the Markov chain obtained from $\MDP_\W$ by resolving
nondeterministic choices according to $S$. Since $S$ is memoryless, $\MDP^S_\W$
has at most as many states as $\MDP_\W$, and so, in particular, it is
finite-state.

Let $\mI \trans{\sigma} M$ be a firing sequence such that $\Pi(\sigma)$ is a
path of $\MDP^S_\W$. Since $\W$ is sound, there is a firing sequence $M
\trans{\tau} \mO$. Further, since $\W$ is 1-safe, it has at most $2^n$
reachable markings, where $n$ is the number of places of $\W$, and so $\tau$
can be chosen of length at most $2^n$.

Since $\mI \trans{\sigma \, \tau} \mO$, the sequence $\sigma\, \tau$ is a run
of $\W$. By Lemma~\ref{lem:mazur2} some run of $\W$ compatible with $S$ is
Mazurkiewicz equivalent to $\sigma \, \tau$. Since $\Pi(\sigma)$ is a path of
$\MDP^S_\W$, this run can be chosen of the form $\sigma \, \tau'$. So every
state of the Markov chain $\MDP^S_\W$ is connected to the final state $\mO$ by
a path of length at most $O(2^n)$. Since the weights of the transitions of $\W$
are all positive, the probability to reach the final marking $\mO$ from any
given marking can be bounded away from zero. So the expected number of steps
until state $\mO$ is reached for the first time in $\MDP^S_\W$, i.e.~the
hitting time of $\mO$, is finite (see e.g. Chapter 1 of~\cite{norris}), and so
the expected time is finite.
\end{proof}

\subsection{Missing proofs of Section~\ref{sec:computation}: The abstraction of
the earliest-first scheduler}

We defined the earliest-first scheduler in terms of the starting time of a
conflict set $C$, defined as the earliest time at which all the transitions of
$C$ become enabled. We assumed without proof that all the transitions become
enabled at the same time. We first prove this in Lemma~\ref{lem:start-time},
and then we prove Theorem~\ref{thm:finitememory}.

\subsubsection{Earliest-first sequences}

We first define the earliest starting time $\start(\sigma)$ of the
last transition of a sequence $\sigma$ as
\begin{align*}
    \start(\epsilon) &\defeq 0 & &\text{and} &
    \start(\sigma t) &\defeq \max_{q \in \preset{t}} \mu(\sigma)_q.
\end{align*}
We show that for every transition of a conflict set, its starting time is equal
to the starting time of the conflict set, and so, in particular, all
transitions of a conflict set have the same starting time.

\begin{restatable}{lemma}{lemstarttime}\label{lem:start-time}
    Let $\W$ be a TPWN, $\sigma$ an occurrence sequence of $\W$ and $C$ a
    conflict set enabled at the marking $M$ with $\mI \trans{\sigma} M$.

    Then for all transitions $t \in C$, we have $\start(\sigma t) = \max_{q \in
    \preset{C}} \mu(\sigma)_q$. Especially, all transitions in $C$ have the
    same earliest starting time after $\sigma$.
\end{restatable}
\begin{proof}
    All places $p \in \preset{C}$ are marked at $M$, so we have $\mu(\sigma)_p
    \ge 0$ for all $p \in \preset{C}$. As $\preset{C} = \cup_{t \in C}
    \preset{t}$, we have that
    \begin{align*}
        \max_{q \in \preset{C}} \mu(\sigma)_q
        &= \max_{t \in C, q \in \preset{t}} \mu(\sigma)_q
        = \max_{t \in C} \max_{q \in \preset{t}} \mu(\sigma)_q
        = \max_{t \in C} \start(\sigma t).
    \end{align*}
    We therefore have $\start(\sigma t) \le \max_{q \in \preset{C}}
    \mu(\sigma)_q$ for any $t \in C$, with equality for at least one $t \in C$.

    For contradiction, now assume that $\start(\sigma t_1) < \max_{q \in
    \preset{C}} \mu(\sigma)_q$ for some $t_1 \in C$. Let $t_2 \in C$ be a
    transition with $\start(\sigma t_2) = \max_{q \in \preset{C}}
    \mu(\sigma)_q$. Then $\start(\sigma t_1) < \start(\sigma t_2)$, so there
    must be $q \in \preset{t_2} \setminus \preset{t_1}$ with $0 \le \max_{p \in
    \preset{t_1}} \mu(\sigma)_p < \mu(\sigma)_q$. Let this $q$ be the last
    such place to become marked during $\sigma$, and let $u$ be the last
    transition occurring in $\sigma$ marking $q$ at timestamp $\mu(\sigma)_q$.
    We have $\sigma = \sigma_1 u \sigma_2$ with $q \in \postset{u}$ and
    $\mu(\sigma_1 u)_q = \mu(\sigma)_q$, and for any transition $u'$ in
    $\sigma_2$, we have $\postset{u'} \cap (\preset{t_2} \setminus
    \preset{t_1}) = \emptyset$.

    Let $\sigma_2 = u_1 \ldots u_k$. We now define the following sequence:
    Initially, set $\tau_0 = \sigma_1$ and $M_0$ as the marking $\mI
    \trans{\sigma_1} M_0$. Then, for $1 \le i \le k$ in order, if $u_i$ is
    enabled at $M_{i-1}$ and $\preset{u_i} \cap \preset{u} = \emptyset$, then
    we set $\tau_i = \tau_{i-1} u_i$ and obtain $M_i$ by $M_{i-1} \trans{u_i}
    M_i$.  Otherwise, set $\tau_i = \tau_{i-1}$ and $M_i = M_{i-1}$. We then
    have $\mI \trans{\sigma_1} M_0 \trans{\tau_k} M_k$.

    \subsubsection{Claim} $t_1$ and $u$ are enabled and concurrent at $M_k$,
    and $t_1$ and $t_2$ are enabled at the marking $M'$ given by
    $M_k \trans{u} M'$.

    \medskip

    We have that $u$ is enabled at $M_1$, and no transition of $\tau_k$ removes
    tokens from $\preset{u}$, thus $u$ is also enabled at $M_k$.

    If $t_1$ were not enabled at $M_k$, then it would only become enabled at
    $M$ by a sequence of transitions initially depending on $u$. However then
    $\start(\sigma t_1) \ge \mu(\sigma)_q$, which contradicts our assumption.
    If $t_1$ and $u$ were not concurrent at $M_k$, then for some place $p \in
    \preset{t_1} \cap \preset{u}$, we would have $\start(\sigma t_1) \ge
    \mu(\sigma)_p \ge \max_{q \in \preset{u}} \mu(\sigma_1)_q + \tau(u) =
    \mu(\sigma_1 u)_q = \mu(\sigma)_q > \start(\sigma t_1)$, again a
    contradiction.

    As $u$ is the last transition of $\sigma$ marking a place $q' \in
    \preset{t_2} \setminus \preset{t_1}$, it must enable $t_2$ at $M'$.  This
    concludes the proof of the claim.

    \bigskip

    As $u$ marks $q$ and the net is 1-safe, we either have $M_k(q) = 0$ or
    $M_k(q) = 1$ and $q \in \preset{u}$. With the claim above, we then have
    that $t_1$ and $u$ are concurrent at $M_k$, and $t_2 \not\in C(t_1, M_k
    \setminus \preset{u})$ but $t_2 \in C(t_1, (M_k \setminus \preset{u}) \cup
    \postset{u})$, so the net is not confusion-free, which contradicts it being
    a TPWN. Therefore the assumption $\start(\sigma t_1) < \max_{q \in
    \preset{C}} \mu(\sigma)_q$ was false, which concludes the proof of the
    lemma.
\end{proof}

\medskip

\subsubsection{Correctness of the finite abstraction}

Now we prove Theorem~\ref{thm:finitememory}. First, we show that the
earliest-first scheduler $\gamma$ actually chooses transitions in increasing
starting time.

\begin{restatable}{lemma}{lemearliestschedule}\label{lem:earliest-schedule}
    Let $\W$ be a TPWN and $\sigma = t_1 \ldots t_n$ be a firing sequence of
    $\W$ compatible with $\gamma$.  Then for every $1 \le i < n$, we have
    $\start(t_1 \ldots t_i) \le \start(t_1 \ldots t_i t_{i+1})$, i.e.~all
    transitions occur in $\sigma$ in increasing order of starting time.
\end{restatable}
\begin{proof}
    We proceed by induction on the length $n$ of $\sigma$. If $n \le 1$, then
    $\sigma = \epsilon$ or $\sigma = t$ for some transition $t$, and the claim
    holds trivially.

    Now assume $n \ge 2$. Let $\sigma = \sigma' u t$ for some sequence
    $\sigma'$ and transitions $u,t$. By induction hypothesis, the claim holds
    for $\sigma' u$, so we only have to show $\start(\sigma' u) \le
    \start(\sigma' u t)$. Let $M,M'$ be markings with $\mI \trans{\sigma'} M'
    \trans{u} M$ and therefore $\mI \trans{\sigma' u} M$. We have
    \begin{align*}
        \start(\sigma' u) &= \max_{q \in \preset{u}} \mu(\sigma')_q
        & &\text{and} &
        \start(\sigma' u t) &= \max_{q \in \preset{t}} \mu(\sigma' u)_q.
    \end{align*}

    We consider three cases:
    \begin{itemize}
        \item Case 1: $(\preset{u} \cup \postset{u}) \cap \preset{t} =
            \emptyset$. Then $u$ and $t$ are independent and firing $u$ cannot
            enable $t$, so $t$ is already enabled at $M'$. As $\sigma$ is
            compatible with $\gamma$ and $u$ is also enabled at $M'$, we have
            $u \in \gamma(\sigma')$ and with the definition of $\gamma$ and
            Lemma~\ref{lem:start-time} we get
            \begin{align}
                \start(\sigma' u)
                &= \max_{q \in \preset{u}} \mu(\sigma')_q
                = \min_{C \in \mathcal{C}(M')} \max_{q \in \preset{C}} \mu(\sigma')_q
                \le \max_{q \in \preset{t}} \mu(\sigma')_q.
                \label{eqn:u-min-conflict}
            \end{align}
            By the definition of $\mu$, we have $\mu(\sigma' u)_q =
            \mu(\sigma')_q$ for any $q \not\in \preset{u} \cup \postset{u}$ and
            thus also for any $q \in \preset{t}$. Therefore
            \begin{align*}
            \start(\sigma' u)
            &\stackrel{(\ref{eqn:u-min-conflict})}{\le} \max_{q \in \preset{t}} \mu(\sigma')_q
            = \max_{q \in \preset{t}} \mu(\sigma' u)_q
            = \start(\sigma' u t).
            \end{align*}
        \item Case 2: $\postset{u} \cap \preset{t} \neq \emptyset$. Let $p \in
            \postset{u} \cap \preset{t}$.  With the definition of $\mu$ we have
            \begin{align*}
                = \start(\sigma' u)
                &= \max_{q \in \preset{u}} \mu(\sigma')_q
                \le \max_{q \in \preset{u}} \mu(\sigma')_q + \td(u) \\
                &= \mu(\sigma' u)_p
                \le \max_{q \in \preset{t}} \mu(\sigma' u)_q
                = \start(\sigma' u t).
            \end{align*}
        \item Case 3: $\postset{u} \cap \preset{t} = \emptyset$ and $\preset{u}
            \cap \preset{t} \neq \emptyset$. Let $p \in \preset{u} \cap
            \preset{t}$. As $t$ is enabled at $M$, we have $M(p) \ge 1$.
            However as $p \not\in \postset{u}$, we have $M'(p) = M(p) + 1 \ge
            2$, which contradicts our assumption that the net is 1-safe.
    \end{itemize}
\end{proof}

\thmfinitememory*
\begin{proof}
    We first give the definitions of $f$, $\nu$ and $r$, and also define
    $\gamma': \Hbot^P \to 2^T$ as the scheduler using the abstraction as used
    in~(\ref{prop1}).
\begin{align*}
    \left(\vec{x} \ominus n\right)_p &\defeq
    \begin{cases}
        \max(\vec{x}_p - n, 0) & \text{if }\vec{x}_p \neq \bot \\
        \bot & \text{if }\vec{x}_p = \bot \\
    \end{cases} &
    f(\vec{x}, t) &\defeq upd(\vec{x}, t) \ominus \max_{p \in \preset{t}} \vec{x}_p \\
    \nu(\epsilon) &\defeq \mu(\epsilon) \text{ and }
    \nu(\sigma t) \defeq \mu(\sigma t) \ominus \max_{p \in \preset{t}} \mu(\sigma)_p  &
    r(\vec{x}) &\defeq \min_{C \in \mathcal{C}(\supp{\vec{x}})}\max_{p \in \preset{C}} \vec{x}_p \\
    \gamma'(\vec{x}) &= \argmin_{C \in \mathcal{C}(\supp{\vec{x}})} \, \max_{p \in \preset{C}} \, \vec{x}_p \\
\end{align*}
From the definitions of $\nu$ and $\start$ it is easy to see that for any
sequence $\sigma$, we have $\nu(\sigma) = \mu(\sigma) \ominus \start(\sigma)$.

We prove the three propositions~(\ref{prop1}), (\ref{prop2}), (\ref{prop3}) and
the additional following claims:
\begin{align}
    \mI &\trans{\sigma} \supp{\nu(\sigma)} \label{prop4} \\
    \start(\sigma) &= \sum_{k=0}^{n-1} r(\nu(t_1 \ldots t_k)) \label{prop5} \\
    \nu(\sigma) &\in \Hbot^P \label{prop6}
\end{align}
Claim~(\ref{prop6}) shows that the range of $\nu$ is actually finite.  The
finite range of $f$ directly follows from this and proposition~(\ref{prop2}).

We proceed by induction on the length of the sequence $\sigma = t_1 \ldots
t_n$. If $n=0$, then $\sigma = \epsilon$, so we must have $\preset{t} =
\{i\}$. Further $\nu(\sigma) = \mu(\sigma)$, $\supp{\nu(\sigma)} = \mI$,
$\gamma(\sigma) = \gamma'(\nu(\sigma))$, $\time(\sigma) = \start(\sigma) = 0$
and as $\max_{p \in \preset{t}} \mu(\sigma) = 0$, we have
\begin{align*}
\nu(\sigma t) &= \mu(\sigma t) \ominus \max_{p \in \preset{t}} \mu(\sigma)
= upd(\mu(\sigma), t) = upd(\nu(\sigma), t) \ominus \max_{p \in \preset{t}} \nu(\sigma) = f(\nu(\sigma), t).
\end{align*}

Now assume $n > 0$, and let $\sigma = \sigma' u$ for some sequence $\sigma'$
and transition $u$. As $\sigma$ is compatible with $\gamma$, we have $u \in
\gamma(\sigma')$. Let $M',M$ be the markings with $\mI \trans{\sigma'} M'
\trans{u} M$. By induction hypothesis, we have $\mI \trans{\sigma'}
\supp{\nu(\sigma')} = M'$. By the definition of $f$ and $upd$, we then have
$\supp{\nu(\sigma)} = \supp{f(\nu(\sigma'), u)} = M$, so $\mI \trans{\sigma}
\supp{\nu(\sigma)}$. This shows~(\ref{prop4}).

Let $s \in \preset{u}$ be a place such that $\mu(\sigma')_s = \max_{p \in
\preset{u}} \mu(\sigma')_p$. By Lemma~\ref{lem:earliest-schedule} we have
\begin{align}
    \start(\sigma')
    \le \start(\sigma' u)
    = \start(\sigma)
    = \max_{q \in \preset{u}} \mu(\sigma')_q
    = \mu(\sigma')_s.
    \label{eqn:start-xs}
\end{align}
As $\sigma$ is compatible with $\gamma$ we have $u \in \gamma(\sigma') \in
\mathcal{C}(M')$. We can derive
\begin{align}
\begin{aligned}
    r(\nu(\sigma'))
    &= \min_{C \in \mathcal{C}(M')} \max_{q \in \preset{C}} \nu(\sigma')_q
    = \min_{C \in \mathcal{C}(M')} \max_{q \in \preset{C}} \left(\mu(\sigma') \ominus \start(\sigma')\right)_q \\
    &= \min_{C \in \mathcal{C}(M')} \max_{q \in \preset{C}} \max(
    \underbrace{\mu(\sigma')_q - \start(\sigma')}_{\ge 0\text{ for }q=s \in \preset{u}}, 0) \\
    &= \min_{C \in \mathcal{C}(M')} \max_{q \in \preset{C}} \mu(\sigma')_q - \start(\sigma') \\
    &= \max_{q \in \preset{t}} \mu(\sigma')_q - \start(\sigma')
    = \start(\sigma) - \start(\sigma').
\end{aligned}
\label{eqn:start-beta}
\end{align}
Therefore
\begin{align*}
    \sum_{k=0}^{n-1} r(\nu(t_1 \ldots t_k))
    &= r(\nu(\sigma')) + \sum_{k=0}^{n-2} r(\nu(t_1 \ldots t_k)) \\
    &\stackrel{\text{I.H.}}{=} r(\nu(\sigma')) + \start(\sigma')
    \stackrel{(\ref{eqn:start-beta})}{=} \start(\sigma') + \start(\sigma) - \start(\sigma')
    = \start(\sigma)
\end{align*}
which shows~(\ref{prop5}).

We have $\time(\sigma) = \max_{p \in P} \mu(\sigma)_p$. For some $p \in P$, we
must have $\mu(\sigma)_p = \time(\sigma) \ge \start(\sigma) \ge 0$. Further,
we have
\begin{align*}
    \max_{p \in P} \nu(\sigma)_p + \sum_{k=0}^{n-1} r(\nu(t_1 \ldots t_k))
    &= \max_{p \in P} \nu(\sigma)_p + \start(\sigma) \\
    &= \max_{p \in P} \left(\mu(\sigma) \ominus \start(\sigma)\right)_p + \start(\sigma) \\
    &= \max_{p \in P} \max(\mu(\sigma)_p - \start(\sigma), 0) + \start(\sigma) \\
    &= \max_{p \in P} \mu(\sigma)_p - \start(\sigma) + \start(\sigma) \\
    &= \max_{p \in P} \mu(\sigma)_p = \time(\sigma)
\end{align*}
which shows~(\ref{prop3}).

By Lemma~\ref{lem:start-time} and~\ref{lem:earliest-schedule}, we have that
\begin{align}\label{eqn:maxstartu}
0 \le \max_{p \in \preset u} \mu(\sigma')_p
    = \start(\sigma' u)
    = \start(\sigma) \le \max_{q \in \preset{C}} \mu(\sigma)_q
\end{align}
for any conflict set $C$ enabled at $M$. We then have
\begin{align*}
    \gamma'(\nu(\sigma))
    &\stackrel{(\gamma')}{=} \argmin_{C \in \mathcal{C}(\supp{\nu(\sigma)})} \max_{q \in \preset{C}} \nu(\sigma)_q
    \stackrel{(\ref{prop4})}{=} \argmin_{C \in \mathcal{C}(M)} \max_{q \in \preset{C}} \nu(\sigma)_q \\
    &\stackrel{(\nu)}{=} \argmin_{C \in \mathcal{C}(M)} \max_{q \in \preset{C}} \left( \mu(\sigma) \ominus \max_{p \in \preset{u}} \mu(\sigma')_p \right)_q \\
    &\stackrel{(\ominus)}{=} \argmin_{C \in \mathcal{C}(M)} \max_{q \in \preset{C}} \max(
    \underbrace{\mu(\sigma)_q - \max_{p \in \preset{u}} \mu(\sigma')_p}_{\ge 0\text{ for some $q \in \preset{C}$ by (\ref{eqn:maxstartu})}}, 0 ) \\
    &= \argmin_{C \in \mathcal{C}(M)} \max_{q \in \preset{C}} \mu(\sigma)_q \stackrel{(\gamma)}{=} \gamma(\sigma).
\end{align*}
which shows~(\ref{prop1}).

We now proceed to show~(\ref{prop2}). The transition $t$ is enabled by
$\sigma$, so we have $M \trans{t} M_t$ for some marking $M_t$. We claim:
\begin{align}\label{eqn:startt}
    \max_{p \in \preset{u}} \mu(\sigma')_p
  = \start(\sigma) \le \start(\sigma t)
  = \max_{p \in \preset{t}} \mu(\sigma)_p
\end{align}
In the case that $t$ was already enabled by $\sigma'$, then $\start(\sigma' t)
\ge \start(\sigma' u)$, as otherwise the scheduler $\gamma$ would have selected
some conflict set starting at $\start(\sigma' t)$ instead of one containing $u$
startig at $\start(\sigma' u)$. Therefore $\start(\sigma t) \ge \start(\sigma'
t) \ge \start(\sigma' u) = \start(\sigma)$. In the other case, $t$ was enabled
by $u$ through some $p \in \postset{u} \cap \preset{t}$, then by the definition
of $\mu$ and $\start$ we get $\start(\sigma) \le \start(\sigma t)$.

We use the following fact about $\ominus$. For any $\vec{x} \in \Nbot^P$ and
$a,b \in \N$, we have:
\begin{align}\label{eqn:ominus-plus}
    \left( \vec{x} \ominus a \right) \ominus b
  = \vec{x} \ominus \left( a + b \right)
\end{align}

Now let $q$ be any place. We prove that $\nu(\sigma t)_q = f(\nu(\sigma),
t)_q$ by case distinction.
\begin{itemize}
    \item Case 1: $q \notin \preset{t} \cup \postset{t}$ and $\nu(\sigma)_q =
        \bot$. Then $M(q) = 0$ and $M_t(q) = 0$, so also $\mu(\sigma) =
        \mu(\sigma t) = \bot$ as well as $\nu(\sigma t) = \bot$. We have:
        \begin{align*}
            f(\nu(\sigma), t)_q
            &\stackrel{(f)}{=} \left( upd(\nu(\sigma), t) \ominus \max_{p \in \preset{t}} \nu(\sigma)_p \right)_q \\
            &\stackrel{(upd)}{=} \left( \nu(\sigma) \ominus \max_{p \in \preset{t}} \nu(\sigma)_p \right)_q
            \stackrel{(\ominus)}{=} \bot = \nu(\sigma t)_q.
        \end{align*}
    \item Case 2: $q \notin \preset{t} \cup \postset{t}$ and $\nu(\sigma)_q
        \neq \bot$. Then also Then $M(q) = 1$ and $M_t(q) = 1$, so also
        $\mu(\sigma)_q = \mu(\sigma t)_q \neq \bot$.  We have:
        \begin{align*}
            f(\nu(\sigma), t)_q
            &\stackrel{(f)}{=} \left( upd(\nu(\sigma), t) \ominus \max_{p \in \preset{t}} \nu(\sigma)_p \right)_q
            \stackrel{(upd)}{=} \left( \nu(\sigma) \ominus \max_{p \in \preset{t}} \nu(\sigma)_p \right)_q \\
            &\stackrel{(\nu)}{=} \left( \left( \mu(\sigma) \ominus \max_{r \in \preset{u}} \mu(\sigma')_r \right) \ominus
                \max_{p \in \preset{t}} \left( \mu(\sigma) \ominus \max_{r \in \preset{u}} \mu(\sigma')_r \right)_p \right)_q \\
            &\stackrel{(\ominus)}{=} \left( \left( \mu(\sigma) \ominus \max_{r \in \preset{u}} \mu(\sigma')_r \right) \ominus
                \max_{p \in \preset{t}} \max\left( \underbrace{\mu(\sigma)_p - \max_{r \in \preset{u}} \mu(\sigma')_r}_{\text{$\ge 0$ for some $p \in \preset{t}$ (\ref{eqn:startt})}}, 0 \right) \right)_q \\
            &= \left( \left( \mu(\sigma) \ominus \max_{r \in \preset{u}} \mu(\sigma')_r \right) \ominus
            \left( \max_{p \in \preset{t}} \mu(\sigma)_p - \max_{r \in \preset{u}} \mu(\sigma')_r \right) \right)_q \\
            &\stackrel{(\ref{eqn:ominus-plus})}{=} \left( \mu(\sigma) \ominus \left( \max_{r \in \preset{u}} \mu(\sigma')_r +
            \left( \max_{p \in \preset{t}} \mu(\sigma)_p - \max_{r \in \preset{u}} \mu(\sigma')_r \right) \right) \right)_q \\
            &= \left( \mu(\sigma) \ominus \max_{p \in \preset{t}} \mu(\sigma)_p \right)_q
            = \left( \mu(\sigma t) \ominus \max_{p \in \preset{t}} \mu(\sigma)_p \right)_q
            \stackrel{(\nu)}{=} \nu(\sigma t)_q
        \end{align*}
    \item Case 3: $q \in \postset{t}$. We have:
        \begin{align*}
            f(\nu(\sigma), t)_q
            &\stackrel{(f)}{=} \left( upd(\nu(\sigma), t) \ominus \max_{p \in \preset{t}} \nu(\sigma)_p \right)_q \\
            &\stackrel{(upd),(\ominus)}{=} \max\left( \max_{p \in \preset t} \nu(\sigma)_p + \tau(t) - \max_{p \in \preset{t}} \nu(\sigma)_p, 0 \right) \\
            &= \max(\tau(t), 0) = \tau(t) \\
            &= \max\left( \max_{p \in \preset t} \mu(\sigma)_p + \tau(t) - \max_{p \in \preset{t}} \mu(\sigma)_p, 0 \right) \\
            &\stackrel{(upd),(\ominus)}{=} \left( upd(\mu(\sigma), t) \ominus \max_{p \in \preset{t}} \mu(\sigma)_p \right)_q \\
            &\stackrel{(\mu)}{=} \left( \mu(\sigma t) \ominus \max_{p \in \preset{t}} \mu(\sigma)_p \right)_q
            \stackrel{(\nu)}{=} \nu(\sigma t)_q
        \end{align*}
    \item Case 4: $q \in \preset{t} \setminus \postset{t}$. Then
        $upd(\nu(\sigma), t)_q = \mu(\sigma t)_q = \bot$. We have:
        \begin{align*}
            f(\nu(\sigma), t)_q
            &\stackrel{(f)}{=} \left( upd(\nu(\sigma), t) \ominus \max_{p \in \preset{t}} \nu(\sigma)_p \right)_q \\
            &= \bot
            = \left( \mu(\sigma t) \ominus \max_{p \in \preset{t}} \mu(\sigma)_p \right)_q
            \stackrel{(\nu)}{=} \nu(\sigma t)_q
        \end{align*}
\end{itemize}
This concludes the proof of~(\ref{prop2}).

Finally, we show claim~(\ref{prop6}). Let $q \in P$. By the induction
hypothesis, we have $\nu(\sigma') \in \Hbot^P$. We show $\nu(\sigma)_q =
f(\nu(\sigma'), u) \in \Hbot$ by case distinction.
\begin{itemize}
    \item Case 1: $q \not\in \preset{u} \cup \postset{u}$. We have:
        \begin{align*}
            \nu(\sigma)_q
            &= f(\nu(\sigma'), u)_q = \left( upd(\nu(\sigma'), u) \ominus \max_{p \in \preset{u}} \nu(\sigma') \right)_q \\
            &\stackrel{(upd)}{=} \left( \nu(\sigma') \ominus \max_{p \in \preset{u}} \nu(\sigma') \right)_q
            \stackrel{(\ominus)}{\le} \nu(\sigma')_q
        \end{align*}
        As $\nu(\sigma)_q \le \nu(\sigma')_q \in \Hbot$, we have $\nu(\sigma)_q \in \Hbot$.
    \item Case 2: $q \in \postset{u}$. We have:
        \begin{align*}
            \nu(\sigma)_q
            &\stackrel{(\nu)}{=} \left( \mu(\sigma) \ominus \max_{p \in \preset{u}} \mu(\sigma')_p \right)_q \\
            &\stackrel{(\mu)}{=} \max \left( \max_{p \in \preset{u}} \mu(\sigma')_p + \tau(t) - \max_{p \in \preset{u}} \mu(\sigma')_p, 0 \right) \\
            &= \max( \tau(t), 0 ) = \tau(t) \in H \subseteq \Hbot
        \end{align*}
        Therefore $\max_{p \in P} Y_p \le \max_{t \in T} \td(t)$.
    \item Case 3: $q \in \preset{u} \setminus \postset{u}$. Then $\nu(\sigma)_q
        = \bot \in \Hbot$.
\end{itemize}
This concludes the induction and the proof of the theorem.

\end{proof}

\subsection{Missing proofs of Section~\ref{sec:lowerbound}: Computing the
expected time is \sharpp{}-hard}

\lemreductionpartone*
\begin{proof}
    First, we fix a total topological order $\preceq$ on $V \cup E$ of $\PN$,
    i.e.~for $x,y \in V \cup E$, if $x \prec y$ then there is no path from $y$
    to $x$. Then fix a scheduler $S$ for $\MDP_\W$ that always chooses an
    enabled transition that is minimal w.r.t $\preceq$ when comparing the
    corresponding vertices or edges, i.e.~for $t_v$ the $v \in V$ and for
    $t_{e,0}$ or $t_{e,1}$ the $e \in E$. From now on consider the Markov
    chain $\MDP_\W^S$.

    Let $E = \set{ e_1, \ldots, e_n }$ be the set of edges of the PERT network.
    A valuation of the random variables $\X = (X_{e_1}, \ldots, X_{e_n})$ is
    given by a vector $\vec{x} = (x_{e_1}, \ldots, x_{e_n}) \in \set{0,1}^n$
    and the probability of $\vec{x}$ is given by $\Pr(\X = \vec{x}) =
    \Pi_{i=1}^n \Pr(X_{e_i} = x_{e_i})$. Let the project duration of $\PN$
    with the valuation $\vec{x}$ be given by $PD(\PN, \vec{x}) = \max_{\pi \in
    \Pi} \sum_{e \in \pi} x_e$.

    \begin{claim}
    There is a bijective function $f : \set{0,1}^n \mapsto {\it Paths}^S$ such
    that for every $\vec{x} \in \set{0,1}^n$:
    (a) $\Pr(\X = \vec{x}) = \mathit{Prob}^S(f(\vec{x}))$, and
    (b) $PD(\PN, \vec{x}) = \time(\mO, f(\vec{x}))$.
    \end{claim}

    \begin{claimproof}
    Let $\pi$ be an infinite path in $\MDP_\W^S$ starting at $\mI$. As the
    workflow net is sound, with probability 1 $\pi$ eventually visits $\mO$
    and then loops. In $\PN$, since every vertex or node $x \in V \cup E$ is
    on a path from $s$ to $t$, we have $s \preceq x \preceq t$. So for every
    $v \in V$, $\pi$ contains $t_v$ exactly once, and for every $e \in E$,
    $\pi$ contains either $t_{e,0}$ or $t_{e,1}$, but not both. Denote the
    transition in $\pi$ by $t_e$. By our choice of scheduler, we then have that
    for any $x,y \in V \cup E$, $t_x$ appears before $t_y$ in $\pi$ if and only
    if $x \prec y$. So the order and positions of $t_x,t_y$ is the same for
    all paths reaching $\mO$, the paths only differ in the choice of $t_e$,
    i.e.~in whether $t_e = t_{e,0}$ or $t_e = t_{e,1}$ holds. Further, for any
    combination of choices for each $t_e$, there is a path containing these
    transitions.

    Define $f$ such that $f(\vec{x}) = \pi$ if{}f $\pi$ is the path such that
    for every edge $e \in E$, if $x_e = 0$, then $\pi$ contains $t_{e,0}$, and
    if $x_e = 1$, then $\pi$ contains $t_{e,1}$. By the argument above, $f$ is
    bijective. We prove (a) and (b).

    \smallskip

    \noindent (a) For the probabilities of $f(\vec{x})$ and $\vec{x}$, we have:
    \begin{align*}
        \Pr(\X = \vec{x})
            &= \prod_{i=1}^n \Pr(X_{e_i} = x_{e_i})
             = \prod_{e \in E} \Pr(X_e = x_e) \\
            &= \prod_{\set{e \in E \mid x_e = 0}} \Pr(X_e = 0) \cdot
               \prod_{\set{e \in E \mid x_e = 1}} \Pr(X_e = 1) \\
            &= \prod_{\set{e \in E \mid x_e = 0}} (1 - p_i) \cdot
               \prod_{\set{e \in E \mid x_e = 1}} p_i \\
            &= \prod_{\set{e \in E \mid t_{e,0} \in f(\vec{x})}} (1 - p_e) \cdot
               \prod_{\set{e \in E \mid t_{e,1} \in f(\vec{x})}} p_e
            = {\it Prob}^S(f(\vec{x}))
    \end{align*}

    \noindent (b) We have $PD(\PN, \vec{x}) = max_{\pi \in \Pi}\sum_{e \in \pi}
    x_e$, where $\Pi$ is the set of paths from $s$ to $t$. This value is
    defined by the longest path from $s$ to $t$. For a directed acylic graph
    with source $s$ and sink $t$, where each edge $e$ has distance $x_e$, it is
    well known that the longest path from $s$ to $t$ is the value of $y_t$ in
    the unique set $\{y_v \mid v \in V\}$ satisfying the equations:
    \begin{align*}
        y_s &= 0 \\
        y_v &= \max\set{ y_u + x_e \mid e = (u,v) \in E } & \text{for every $v \in V \setminus \set{s}$}
    \end{align*}
    \noindent So we have $PD(\PN, \vec{x}) = y_t$. Let $\sigma$ be the
    transition sequence corresponding to $f(\vec{x})$ in $\W$. We have $\mI
    \trans{\sigma} \mO$, and $\sigma$ contains the transition $t_v$ for each
    vertex $v \in V$ and one of $t_{e,0}$ or $t_{e,1}$ for each edge $e \in E$.
    As before, for every $e \in E$, denote by $t_e$ the $t_{e,0}$ or $t_{e,1}$
    occurring in $\sigma$.

    Note that every place of $\W$ becomes marked exactly once during the
    occurrence of $\sigma$. For any place $p \in P$, denote by $\time(p)$ the
    time at which $p$ becomes marked, as given by $\mu(\sigma')_p$ for the
    timestamp function $\mu$ and the prefix $\sigma'$ of $\sigma$ after which
    $p$ becomes marked. Additionally, for any vertex $v \in V$, all places of
    the form $[v,e]$ become marked at the same time. Denote this time by
    $\time(v)$.

    From the definition of the timestamp and time functions and the
    construction of $\W$, we have that for any vertex $v \in V \setminus
    \set{s}$ and for any edge $e = (u,v) \in E$
    \begin{align*}
        \time(v) &= \max\set{ \time([e,u]) \mid e = (u,v) \in E } \\
        \time([e,u]) &= \time([u,e]) + \tau(t_e) = \time(u) + \tau(t_e)
    \end{align*}
    which combined give us
    \begin{align*}
        \time(v) &= \max\set{ \time(u) + \tau(t_e) \mid e = (u,v) \in E }.
    \end{align*}
    Further, we have $\time(s) = \time(i) = 0$, $\time(o) = \time(t)$ and, by
    the definition of the function $f$, $\tau(t_e) = x_e$. This gives us the
    following set of equations:
    \begin{align*}
        \time(s) &= 0 \\
        \time(v) &= \max\set{ \time(u) + x_e \mid e = (u,v) \in E } & \text{for every $v \in V \setminus \set{s}$}
    \end{align*}
    These are exactly the same equations as for the project duration, so we
    necessarily have $\time(v) = y_v$ for any vertex $v \in V$.

    Finally, for the time of $f(\vec{x})$, we have:
    \begin{align*}
    \time(\mO, f(\vec{x})) &= \time(\sigma) = \max_{p \in \mO} \mu(\sigma)_p
    = \mu(\sigma)_o = \time(o) = \time(t) = y_t = PD(\PN,\vec{x})
    \end{align*}

    \end{claimproof}

    For the expected project duration, we have:
    \begin{align*}
        \Ex(PD(\PN)) &=
        \sum_{\vec{x} \in \set{0,1}^n} PD(\PN, \vec{x}) \cdot \Pr(\X = \vec{x})
    \end{align*}

    For the expected time, as the net is sound and acyclic, we have:
    \begin{align*}
        ET_{\W_\PN} &= ET_{\W_\PN}^S = ET_{\W_\PN}^S(\mO)
              = \sum_{\pi \in \mathit{Paths}^S} \time(\mO,\pi) \cdot \mathit{Prob}^S(\pi)
    \end{align*}

    With the help of the function $f$ of the claim, we get $ET_{\W_\PN} =
    \Ex(PD(\PN, \X))$.
\end{proof}

\subsubsection{Replacing rational weights by weights 1}

We now formally describe how to replace all rational weights $p_e$,
$\overline{p_e}$ in $\W_\PN$ by weights $1$. For each edge $e = (u,v) \in E$,
apply the following construction:

\begin{itemize}
    \item Remove the transitions $t_{e,0}$ and $t_{e,1}$.
    \item For each $1 \le i \le k$, add a place $q_{e,i}$.
    \item For each $1 \le i \le k$, add two transitions $a_{e,i}, b_{e,i}$,
        both with arcs from $q_{e,i}$, an arc from $a_{e,i}$ to $[e,v]$, and if
        $i < k$, an arc from $b_{e,i}$ to $q_{i+1}$, otherwise an arc from
        $b_{e,i}$ to $[e,v]$, Assign $w(a_{e,i}) = w(b_{e,i}) = 1$,
        $\tau(a_{e,i}) = p_i$ and $\tau(b_{e,i}) = 0$.
    \item Add a transition $a_{e,0}$ with an arc from $[u,e]$ and, if $k \ge
        1$, an arc to $q_1$, otherwise an arc to $[e,v]$. Assign $w(a_{e,0}) =
        1$ and $\tau(a_{u,0}) = p_0$.
\end{itemize}

We show that that above construction preserves the expected time.

\begin{restatable}{lemma}{lemreductionparttwo}\label{lem:reduction-part-two}
    Let $\W = \W_\PN$ be a workflow net corresponding to a two-state stochastic
    PERT network $\PN$. Let $\W'$ be the result of replacing every gadget of
    $\W$ of the form described in Figure~\ref{fig:reduction-weight-rational} by
    a gadget of the form described in Figure~\ref{fig:reduction-weight-one},
    following the construction rules above.  Then $ET_{\W'} = ET_\W$.
\end{restatable}
\begin{proof}
    Let $\Net$ be the subnet of $\W$ from
    Figure~\ref{fig:reduction-weight-rational} and $\Net'$ be the corresponding
    subnet of $\W'$ of the form from Figure~\ref{fig:reduction-weight-one}. We
    fix schedulers $S,S'$ for $\MDP_{\W}$ and $\MDP_{\W'}$ that for each edge
    $e \in E$ choose transitions of $e$ (i.e.~$t_{e,0}, t_{e,1}, a_{e,i},
    b_{e,i}$) consecutively before choosing any transition for a vertex $v$ or
    another edge $e'$. It is easy to to see that in any path of
    $\MDP_{\W}^{S}$ and $\MDP_{\W'}^{S'}$, for each edge $e = (u,v)$, whenever
    the first transition for $e$ occurs, it moves the token from $[u,e]$, then
    all transitions of $e$ occur, and a token is placed into $[e,v]$ without
    any other transitions in between.

    We regard the subnets $\Net,\Net'$ as workflow nets themselves with initial
    place $i = [u,e]$ and final place $o = [e,v]$. It suffices to show that the
    time distributions for both nets are the same. In each net, we can choose a
    transition with time $1$ at most once, therefore we only need to show that
    the probability of choosing a transition with time $1$ is equal in $\Net$
    and $\Net'$.

    In $\Net$, the probability of choosing $t_{e,1}$ is $p_e$. We assumed that
    $p_e = \sum_{i=0}^k 2^{-i} p_i$ for $p_i \in \set{0, 1}$. In $\Net'$, the
    probability of choosing $a_{e,0}$ is $1 = 2^0$, and for $1 \le i \le k$,
    the probabilty of choosing $a_{e,i}$ is equal to $2^{-i}$, as we need to
    choose all $b_{e,j}$ with $j < i$ before, and then choose $a_{e,i}$ instead
    of $b_{e,i}$. For all $i$, we have $\tau(a_{e,i}) = 1$ iff $p_i = 1$.
    Therefore the total probability of choosing a transition with time $1$ in
    $\Net'$ is $\sum_{i=0}^k 2^{-i} p_i = p_e$.
\end{proof}

\subsubsection{Computing the expected time is \sharpp{}-hard}

\thmapproximationsharpphard*
\begin{proof}
In~\cite{Hagstrom88}, it is only shown that computing $\Ex(PD(\PN))$ is
\sharpp{}-hard by using the results of~\cite{ProvanB83}. However, it is easy
to see that their construction also shows that computing an
$\epsilon$-approximation for $\Ex(PD(\PN))$ is \sharpp{}-hard.
In~\cite{ProvanB83}, the authors showed that given a graph $G = (V, E)$,
vertices $s,t \in V$, a rational $p$ with $0 \le p \le 1$ and an $\epsilon >
0$, it is \sharpp{}-complete to compute a rational $r$ such that $r - \epsilon
< f(G,s,t;p) < r + \epsilon$, where $f(G,s,t;p)$ is the reliability measure of
the graph. In~\cite{Hagstrom88}, the constructed two-state stochastic PERT
network used to show \sharpp{}-hardness satisfies $f(G,s,t;p) = \Ex(PD(\PN)) -
2$.

Our constructed net $\W_\PN$ satisfies the constraints for the given TPWN, and
by Lemma~\ref{lem:reduction-part-one} and~\ref{lem:reduction-part-two}, it
satisfies $ET_{\W_\PN} = \Ex(PD(\PN))$. Therefore it follows that it is also
\sharpp-hard to compute an $\epsilon$-approximation for $ET_{\W_\PN}$, as this
directly gives an $\epsilon$-approximation for $f(G,s,t;p) = ET_{\W_\PN} - 2$.
\end{proof}

\end{document}